\documentclass[aps,reprint,amsmath,amssymb,superscriptaddress,showkeys]{revtex4-1}
\pdfoutput=1

\usepackage[latin1]{inputenc}
\usepackage{graphicx}
\usepackage{amsmath}
\usepackage{amsfonts}
\usepackage{amsmath}
\usepackage{amssymb}
\usepackage{soul}
\usepackage{url}
\usepackage{bm}
\usepackage{amsthm}
\usepackage{mathtools}
\usepackage{dsfont}
\usepackage{stmaryrd}
\usepackage{adjustbox}
\usepackage{makecell}
\usepackage{xspace}
\usepackage{epigraph}
\usepackage{cancel}
\usepackage{lipsum} 
\usepackage{cancel}
\usepackage{float}
\usepackage{verbatim}
\usepackage{epstopdf}
\usepackage{mathtools}
\usepackage{mathrsfs}
\usepackage{amsfonts}
\usepackage{subfigure}
\usepackage{color}
\usepackage{cancel}
\usepackage{dsfont}
\usepackage[shortlabels]{enumitem}
\usepackage{algorithm}
\usepackage{algpseudocode}
\usepackage{arydshln}
\usepackage[official]{eurosym}
\usepackage[binary-units]{siunitx}

\usepackage{filecontents}
\usepackage[dvipsnames]{xcolor}
\usepackage{hyperref}
\usepackage{cleveref}
\definecolor{BrewerRed}{RGB}{228,26,28}

\newcommand\myshade{85}
\colorlet{mylinkcolor}{BrewerRed}
\colorlet{mycitecolor}{NavyBlue}
\colorlet{myurlcolor}{Aquamarine}

\hypersetup{
  linkcolor  = mylinkcolor!\myshade!black,
  citecolor  = mycitecolor!\myshade!black,
  urlcolor   = myurlcolor!\myshade!black,
  colorlinks = true,
}

\let\emptyset\varnothing

\newcommand\independent{\protect\mathpalette{\protect\independenT}{\perp}}
\def\independenT#1#2{\mathrel{\rlap{$#1#2$}\mkern2mu{#1#2}}}

\makeatletter
\newcommand{\subalign}[1]{%
  \vcenter{%
    \Let@ \restore@math@cr \default@tag
    \baselineskip\fontdimen10 \scriptfont\tw@
    \advance\baselineskip\fontdimen12 \scriptfont\tw@
    \lineskip\thr@@\fontdimen8 \scriptfont\thr@@
    \lineskiplimit\lineskip
    \ialign{\hfil$\m@th\scriptstyle##$&$\m@th\scriptstyle{}##$\hfil\crcr
      #1\crcr
    }%
  }%
}
\makeatother

\newcommand{\Corder}{\preceq_{\text{c}}}
\newcommand\PI[1]{\ensuremath{S_{\partial}^{#1}}\xspace}

\newcommand{\bs}[1]{\boldsymbol{#1}}

\newcommand\bX{\ensuremath{\bm X}\xspace}

\definecolor{darkgreen}{RGB}{50,200,8}

\renewcommand{\bs}[1]{\boldsymbol{#1}}

\graphicspath{ {./figs/} }
\newcommand\figsubref[2]{\hyperref[#1]{\ref*{#1}#2}}

\newtheorem{definition}{Definition}
\newtheorem{theorem}{Theorem}
\newtheorem{lemma}{Lemma}
\newtheorem{proposition}{Proposition}

\newtheorem{example}{Example}

\usepackage{tikz-network}

\usepackage{pgfplots}
\pgfplotsset{compat = newest}

\usetikzlibrary{shapes.misc, positioning}
\usetikzlibrary{backgrounds,bayesnet,matrix}

\usepgfplotslibrary{fillbetween}
\usepgfplotslibrary{colorbrewer}

\usetikzlibrary{pgfplots.external}
\usetikzlibrary{external}
\definecolor{BrewerGreen}{RGB}{77,175,74}
\definecolor{BrewerBlue}{RGB}{55,126,184}
\definecolor{BrewerPurple}{RGB}{152,78,163}
\definecolor{BrewerOrange}{RGB}{255,127,0}
\definecolor{BrewerYellow}{RGB}{215,215,41}

\usepackage{filemod}
\newcommand{\includetikz}[2]{\includegraphics{#2}}
\usepackage{booktabs}
\usepackage{ctable}
\newcommand{\ra}[1]{\renewcommand{\arraystretch}{#1}}

\pgfplotsset{grid style={dashed,gray!40}}

\pgfplotsset{
  /pgfplots/scatter legend/.style={
    /pgfplots/legend image code/.code={\draw[##1,mark size=0.9pt,yshift=-0.1em, xshift=7pt] plot coordinates {
      (-0.1em, 0.1em)
      (0.2em, 0.45em)
      (0.3em, -0.05em)
      (0.6em, 0.3em)
    };},
  },
}

\definecolor{lightblue}{RGB}{166,206,227}
\definecolor{darkblue}{RGB}{31 ,120,180}
\definecolor{lightred}{RGB}{251,154,153}
\definecolor{darkred}{RGB}{227,26,28}

\begin{document}

\title{An operational information decomposition via synergistic disclosure}
\author{Fernando Rosas} 
\thanks{F.R. and P.M. contributed equally to this work.\\E-mail: f.rosas@imperial.ac.uk, pam83@cam.ac.uk}
\affiliation{Data Science Institute, Imperial College London, London SW7 2AZ}
\affiliation{Center for Psychedelic Research, Department of Medicine, Imperial College London, London SW7 2DD}
\affiliation{Center for Complexity Science, Imperial College London, London SW7 2AZ}

\author{Pedro Mediano}
\thanks{F.R. and P.M. contributed equally to this work.\\E-mail: f.rosas@imperial.ac.uk, pam83@cam.ac.uk}
\affiliation{Department of Psychology, University of Cambridge, Cambridge CB2 3EB}

\author{Borzoo Rassouli}
\affiliation{School of Computer Science and Electronic Engineering, University of Essex, Colchester CO4 3SQ}

\author{Adam Barrett}
\affiliation{Sackler Center for Consciousness Science, Department of Informatics, University of Sussex, Brighton BN1 9RH}


\begin{abstract}

Multivariate information decompositions hold promise to yield insight into
complex systems, and stand out for their ability to identify synergistic phenomena. 
However, the adoption of these approaches has been hindered by there being
multiple possible decompositions, and no precise guidance for preferring one over
the others. 
At the heart of this disagreement lies the absence of a clear operational 
interpretation of what synergistic information is. 
Here we fill this gap by proposing a new information decomposition 
based on a novel operationalisation of informational synergy, which
leverages recent developments in the literature of data privacy. 
Our decomposition is defined for any number of information sources, and its 
atoms can be calculated using elementary optimisation techniques. 
The decomposition provides a natural coarse-graining that scales
gracefully with the system's size, and is applicable in a wide range of
scenarios of practical interest. 

\end{abstract}

\maketitle

\section{Introduction}

The familiarity with which we relate to the notion of ``information'' -- due to its central role in our modern worldview -- is at odds with the mysteries still surrounding some of its fundamental properties. 
One such mystery is the nature and role of \textit{synergistic information}, which is present in systems that exhibit global interdependencies that are not traceable from any of their subsystems. Synergistic
relationships have shown to be instrumental in a wide range of systems, including the nervous system~\cite{ganmor2011sparse,wibral2017quantifying}, artificial
neural networks~\cite{tax2017partial}, cellular
automata~\cite{rosas2018selforg}, and music
scores~\cite{rosas2019quantifying}. 
Furthermore, the concept of synergy traces a particularly
promising road to formalise the notion of ``the whole being greater than the sum of
the parts'', one of the long-standing aims of complexity science~\cite{waldrop1993complexity}.

Informational synergy has been studied following various approaches, 
including redundancy-synergy balances~\cite{chechik2002group,varadan2006computational,barrett2015exploration,rosas2019quantifying}, information geometry~\cite{amari2001information,schneidman2003network}, and others. Within this literature, one of the most elegant and powerful proposals is the \emph{Partial Information Decomposition} (PID) framework~\cite{williams2010nonnegative}, which divides information into \emph{redundant} (contained in every part of the system), \emph{unique} (contained in only one part), and \emph{synergistic} (contained in the whole, but not in any part) components. One peculiarity of the PID framework is the absence of precise prescriptions about how synergy should be quantified~\footnote{In effect, PID merely states formal relationships between its atoms and various Shannon's mutual information terms.}; and despite numerous efforts, an agreed-upon measure of synergy remains elusive~\cite{ince2017measuring,bertschinger2014quantifying,finn2018pointwise,james2019unique}. Most approaches to quantify synergy proceed by postulating axioms encoding some ``intuitive'' desiderata, 
which should ideally lead towards a unique measure -- following the well-known axiomatic derivation of Shannon's entropy~\cite{thurner2018introduction}. 
Unfortunately, a number of critical incompatibilities between some of these axioms have been reported~\cite{rauh2014reconsidering,kolchinsky2019novel}, 
which reveals the limitations of our intuition as a guide 
within the counterintuitive realms of high-order statistics.

Building on these remarks, we argue that measures of synergy with little concrete, 
operational meaning provide a limited advance from mere qualitative criteria. 
Moreover, as argued by Kolchinsky~\cite{kolchinsky2019novel}, there might exist 
not a single but multiple reasonable definitions of synergy, and hence
it is crucial to clarify what each proposed measure is capturing~\cite{feldman1998measures}. 
There have been a few attempts to formulate operational measures of redundant~\cite{ince2017measuring,finn2018pointwise} and unique information~\cite{bertschinger2014quantifying,banerjee2018computing}, but these
efforts are still in progress, and apply only indirectly to synergy~\footnote{The direct implications of these approaches are about the unique information, and apply to the synergy only via additional equalities with mutual information terms.}. Providing a clear operational meaning for synergy is, to the best of our knowledge, an important unresolved challenge.

In this paper we 
put forward a \emph{synergy-centered information decomposition}, rooted on the notion of \emph{synergistic data disclosure} from the literature of data privacy~\cite{rassouli2018latent,Rassouli2019} and synergistic variables introduced in Ref.~\cite{quax2017quantifying}. In this decomposition, synergy corresponds to the information that can be disclosed about a system without revealing the state of any of its parts. This measure is efficiently computable and is, to the best of our knowledge, the first to provide a direct operational interpretation for synergistic information. 
Moreover, our proposed decomposition is 
applicable to any number of source variables, 
and its operational meaning provides 
natural coarse-grainings that 
enable 
useful tools for 
practical analysis.

The paper is structured as follows. First, Section~\ref{sec2}
introduces our operational definition of synergy, and
Section~\ref{sec:decomp} uses it to build our proposed decomposition. The
decomposition's coarse-graining is discussed in Section~\ref{sec:backbone}, and
the special case of \emph{self-synergy} in
Section~\ref{sec:selfsyn}. Finally, the relationship with
other decompositions is studied in Section~\ref{sec:pid}.

\section{Synergy and data disclosure}\label{sec2}

Our goal is to develop a method to decompose the information that a multivariate system $\bX\coloneqq (X_1,\dots,X_n)$ provides about a target variable
$Y$, as quantified by Shannon's mutual information $I(\bX; Y)$. 
Our approach consists of three steps:
\begin{enumerate}

  \item Introduce \emph{synergistic channels}, which convey information about \bX
  but not about any of its parts (Sec.~\ref{sec:synchan});

  \item Define \emph{synergistic disclosure} as the maximum amount of information 
  about $Y$ that can be obtained through a synergistic channel on \bX
  (Sec.~\ref{sec:fundamental}); and

  \item Build an information decomposition by computing the synergistic disclosure
  for every node of a lattice, and using the M\"obius inversion formula
  (Sec.~\ref{sec:decomp}).

\end{enumerate}

The rest of this section provides technical details about synergistic disclosure, building upon the work recently reported in
Refs.~\cite{rassouli2018latent,Rassouli2019}.

\subsection{Synergistic channels}
\label{sec:synchan}

Consider a system described by $n$ variables, $\bm X\coloneqq
(X_1,\dots,X_n)$, where each $X_k$ takes values on a discrete alphabet
$\mathcal{X}_k$ of cardinality $|\mathcal{X}_k|$. Consider also 
a channel that is applied on $\bm X$ to generate a
scalar observable $V$, which is characterised by a conditional 
distribution $p_{V|\bm X}$. We are interested in a particular class of observables, 
which carry information about $\bm X$ while revealing no information about 
specific subsystems.

Subsystems of $\bm X$ can be represented by sets of indices of the form
$\alpha=\{n_1,\dots,n_k\} \subset [n]$, with $[n] \coloneqq \{1,\dots,n\}$ being a shorthand notation, and the corresponding subsystem being denoted by 
$\bm X^\alpha = (X_{n_1},\dots,X_{n_k})$. 
We consider collections of subsystems, which are represented 
by \emph{source-sets} of the form $\bm \alpha = \{\alpha_1,\dots,\alpha_L\}$, 
where $\alpha_j\subset [n]$ for all $i=1,\dots,L$. For example, 
possible source-sets for $n=2$ are $\{\emptyset\}$, $\{\{1\}\}$, $\{\{1\},\{1,2\}\}$, etc. 
With the notion of source-set in hand, we can formally define synergistic
channels as follows:
\begin{definition}
A channel $p_{V| \bm X}$ is $\bm \alpha$-synergistic for $\bm
\alpha=\{\alpha_1,\dots,\alpha_L\}$, if $V\independent \bm X^{\alpha_i}$,
$\forall i = 1, \dots, L$. The set of all $\bm\alpha$-synergistic channels is
denoted by
\begin{equation}\label{def_set}
\mathcal{C}(\bm X; \bm \alpha) = \bigg\{ p_{V|\bm X} \:\bigg| \:V\independent \bm X^{\alpha_i},\forall i\in[L]\bigg\}.
\end{equation}
A variable $V$ generated via an $\bm\alpha$-synergistic channel is said to be
an $\bm\alpha$-synergistic observable.
\end{definition}
Due to the independence constraints, an $\bm\alpha$-synergistic observable $V$
satisfies $I(\bX^{\alpha_i}; V) = 0$ for all $i=1,\dots,L$. Thus the name synergistic: by construction, 
an $\bm\alpha$-synergistic observable $V$ might convey information about the whole, $\bX$, 
while disclosing no information about the corresponding parts $\bm X^{\alpha_1},\dots,\bm X^{\alpha_L}$.

\begin{example}
If $\bX=(X_1,X_2)$ are two independent fair coins, then the observable $V= X_1\textnormal{\texttt{xor}} X_2$
given by
\begin{equation}
X_1\textnormal{\texttt{xor}} X_2 \coloneqq
\begin{cases}
0 &\text{if } X_1=X_2 \\
1 &\text{if } X_1\neq X_2~.
\end{cases}
\end{equation}
is $\bm\alpha$-synergistic for $\bm\alpha = \{\{1\},\{2\}\}$.
\end{example}

Note that $p_{V|\bm X}$ can be depicted as a rectangular matrix. Elegant algebraic methods for characterising synergistic channels based on this matrix representation are available, and are discussed in Appendix~\ref{app:channel_description}.

\subsection{Fundamental properties of synergistic disclosure}
\label{sec:fundamental}

Now that synergistic channels have been defined, let us 
formulate our measure of synergistic disclosure. To do this, consider 
a target variable $Y$, potentially having some dependence on \bX according to a given 
joint distribution $p_{\bm X,Y}$. We
are interested in quantifying to what extent the collective properties of $\bm
X$ can predict $Y$ without revealing any information about the subsystems
$\bX^{\alpha_1},\dots,\bX^{\alpha_L}$. This intuition can be naturally operationalised by the
mutual information between $Y$ and the $\bm\alpha$-synergistic observables of $\bm X$, as
described in the next definition.
\begin{definition}\label{def:syndisc}
The $\bm \alpha$-synergy between sources $\bm X$ and target $Y$ is defined as
\begin{equation}\label{eq:def_syn}
S^{\bm \alpha}(\bX \rightarrow Y)\coloneqq \sup_{\substack{p_{V|\bm X}\in \mathcal{C}(\bm X;\bm\alpha):\\V-\bm X-Y}} I(V;Y)~.
\end{equation} 
\end{definition}
\noindent Above, the supremum is calculated over all the $\bm\alpha$-synergistic channels $p_{V|\bm X}$, so that the joint distribution of $(V,Y)$ over which $I(V;Y)$ is calculated is of the form 
\begin{equation}
    \mathbb{P}\{V=v,Y=y\} = \sum_{\bm x\in \prod_{i=1}^n \mathcal{X}_i}p_{V|\bm X}(v|\bm x) p_{\bm X,Y}(\bm x,y)~.
    \nonumber
\end{equation}
Additionally, it can be verified that the supremum in \eqref{eq:def_syn} is attained, and hence, it is a maximum~\footnote{To verify this, first one shows that it suffices to have a random variable $V$ with a finite alphabet by means of cardinality bounding techniques. Then, using the fact that any finite probability simplex is a compact set, the supremum in (\ref{eq:def_syn}) has to be attained due to the continuity of the mutual information.}. Finally, this definition can used to extend the notion of synergy over any $f$-information, as discussed in Appendix~\ref{app:finfo}.

Please note that Definition~\ref{def:syndisc} has a concrete operational
interpretation : $S^{\bm\alpha}(\bX \rightarrow Y)$ \textit{represents the amount of information about
$Y$ that can be disclosed from $\bX$ while revealing no information about
$\bX^{\bm\alpha_1},\dots,\bX^{\bm\alpha_L}$}. This operationalisation has its origins in the
context of data privacy scenarios, as discussed in Refs.~\cite{rassouli2018latent,Rassouli2019}. 
Please note that this strongly contrasts with previous approaches to information decomposition,
which have proceeded by writing down an axiomatic base and then formulating a
measure consistent with those axioms. 
As a matter of fact, most problems in information theory are operational in 
nature~\footnote{Shannon himself employed the known formula for entropy because
it had a strict operational meaning as minimum description length, not because
it was the result of the four celebrated axioms.}, and hence one could argue
that this approach lies closer to Shannon's original contribution.

Let us explore some basic properties of our measure of synergy,
$S^{\bm\alpha}$. A first fortunate feature is that this quantity is computable via
simple optimisation techniques, which is a direct extension of Ref.~\cite[Theorem~1]{rassouli2018latent}.
\begin{theorem}\label{teo1}
  The supremum in Eq.~\eqref{eq:def_syn} is always attained, and the corresponding 
synergistic channel can be obtained as the solution to a standard linear-programming problem.
\end{theorem}
While the proof of Theorem~\ref{teo1} is omitted, interested readers can find the corresponding details in Ref.~\cite[Section~3]{Rassouli2019}. Additionally, software alternatives to compute $S^\alpha$ are discussed in Section~\ref{sec:conclusion}.

Despite the guarantees provided by Theorem~\ref{teo1}, it is useful to have simple bounds. 
Note that, due to the data processing inequality,
$S^{\bm\alpha}$ satisfies $S^{\bm\alpha}(\bm X\rightarrow Y) \leq I(\bm X;Y)$
for all $\bm \alpha$. The following result introduces a less trivial upper
bound.

\begin{proposition}\label{pr:upper_bound}
The following upper bound holds for $S^{\bm\alpha}$:
\begin{equation}\label{uppbound}
S^{\bm\alpha }(\bm X \rightarrow Y) \leq \min_{j\in \{1,\dots,L\}} I(Y; \bm X^{-\alpha_j}|\bm X^{\alpha_j}) ,
\end{equation}
where $\bm X^{-\alpha_j} \triangleq \{X_1,\ldots,X_n\}\backslash
\{X_{n_1},\dots,X_{n_k}\}$ with $\alpha_j=\{n_1,\dots,n_k\}$.
\end{proposition}

\begin{proof}
  See Appendix~\ref{app:upper_bound}.
\end{proof}

The above property sometimes provide a shortcut to calculate $S^{\bm\alpha}$, as if
one finds a particular synergistic observable that attains this upper bound
then it is clear that it is maximal. One immediate consequence of this
Proposition, noting that $I(\bm X^{-\alpha_j}; Y|\bm X^{\alpha_j}) = I(\bm X;
Y) - I(\bm X^{\alpha_j}; Y)$, is that
\begin{equation}
  I(\bm X; Y) - S^{\bm\alpha }(\bm X \rightarrow Y) \geq \max_{j\in\{1,\dots,n\}} I(\bm X^{\alpha_j}; Y)~.
\end{equation}
In other words, the amount of non-synergistic information is lower-bounded by
the amount of information carried by the most strongly correlated subgroup.

Further details on $S^{\bm\alpha}$, including properties of its bounds,
algebraic properties, and a data processing inequality, are presented in
Section~\ref{sec:axioms} and Appendix~\ref{app:further_props}.

\section{Information decomposition}
\label{sec:decomp}

This section uses the functional definition of $\bm\alpha$-synergy to
formulate our proposed information decomposition. For this, we focus on the study of
the sets of constraints of the form $\bm\alpha=\{\alpha_1,\dots,\alpha_L\}$,
which are the argument in the synergy $S^{\bm\alpha}(\bX\rightarrow Y)$. For
such sets, we say $|\bm\alpha|\coloneqq L$ is the cardinality of the set.

\subsection{The extended constraint lattice}\label{sec:exteneded_lattice}

Let us start by observing that not all source-sets yield unique synergistic
channels. As a simple example, if $\bm\alpha = \{\{1,2\}\}$ and $\bm\beta = \{\{1\},\{1,2\}\}$ one has that $\mathcal{C}(\bm X;\bm\alpha) = \mathcal{C}(\bm X;\bm\beta)$,
as all the additional constraints in $\bm\beta$ are subsumed by the constraints in $\bm\alpha$.
More formally, we say that
two source-sets are equivalent, denoted by $\bm\alpha\equiv_I\bm\beta$, if
$\mathcal{C}(\bX; \bm\beta) = \mathcal{C}(\bX; \bm\alpha)$. 
Our next result shows that the set of anti-chains
\begin{equation} \label{eq:ordering}
\mathcal{A}^* =\{ \bm\alpha=\{\alpha_1,\dots,\alpha_L\}: \alpha_i \subset [n], \alpha_i\not\subset\alpha_j \forall i\neq j \}
\end{equation}
contains exactly one member of each equivalence class, and this member is the
simplest such source-set.
\begin{lemma}\label{lemma:antichains}
For any $\bm \beta = \{\beta_1,\dots,\beta_M\}$ with $\beta_i\subset [n]$,
there exists one and only one $\bm\alpha\in\mathcal{A}^*$ such that
$\bm\beta\equiv_I\bm\alpha$. Moreover, if $\bm\beta\equiv_I\bm\alpha$ and $\bm\alpha\in\mathcal{A}^*$, then
$|\bm\beta|\geq|\bm\alpha|$.
\end{lemma}
\begin{proof}
See Appendix~\ref{app:antichains}.
\end{proof}

In other words, considering collections of indices that are not anti-chains
would not provide new classes of channels, as broader subunits subsume smaller
ones. This property brings strong reminiscences of Williams and Beer's
redundancy lattice~\cite{williams2010nonnegative} -- which we will discuss in
detail in Section~\ref{sec:pid}~\footnote{In fact, we employ the notation $\mathcal{A}^*$ to
differentiate this set from with their definition of antichains set $\mathcal{A}$ that 
doesn't include the empty set -- which plays an important role in our framework.}.

In addition to the set of nodes, to build a lattice on which one can formulate
a decomposition one needs a partial order relationship. Considering our setup, 
a natural candidate is the order introduced by James \emph{et al.} in their 
proposed \emph{constraint lattice}~\cite{james2018unique}, defined by
\begin{equation}
\bm\alpha \Corder \bm \beta
\iff \forall \alpha \in \bm\alpha, ~ \exists \beta \in \bm\beta: \alpha \subseteq \beta 
\end{equation}
for $\bm\alpha,\bm\beta\in\mathcal{A}^*$. Intuitively, $\bm\alpha \Corder \bm \beta$ 
means that all the constraints 
imposed by
$\bm\alpha$ are included
within those imposed by
$\bm\beta$,
and therefore
$\mathcal{C}(\bX; \bm\beta) \subseteq \mathcal{C}(\bX; \bm\alpha)$.

Putting these structures together generates the \textit{extended constraint
lattice} $\mathcal{L}^* : = (\mathcal{A}^*, \Corder )$, which extends the
lattice introduced by James \emph{et al.}~\cite{james2018unique}, and has
been recently used by Ay \emph{et al.}~\cite{ay2019information}. The cases $n=2$ and $n=3$
are depicted in Fig.~\ref{fig:lattice}. Importantly, in constrast with James'
proposal,
$\mathcal{L}^*$
includes nodes that do not cover all the
sources. The resulting lattice is isomorphic in shape to Williams and Beer's
redundancy lattice, but with different relationships between the nodes. Despite
this similarity, however, comparisons between these two lattices are not
straightforward (c.f. Sec.~\ref{sec:pid}).

The lattice $\mathcal{L}^*$ possesses some interesting properties, most prominently:
\begin{lemma}\label{lemma:monotonic}
If $\bm\alpha, \bm\beta\in\mathcal{L}^*$ and $\bm\alpha \Corder \bm\beta$, then 
\begin{equation}
S^{\bm\alpha}(\bm
X\rightarrow Y) \geq S^{\bm\beta}(\bm X\rightarrow Y)~.
\end{equation}
\end{lemma}
\begin{proof}
  See Appendix~\ref{app:Corder}.
\end{proof}

This result shows that $S^{\bm\alpha}(\bX\rightarrow Y)$ is a non-increasing function of $\bm\alpha\in\mathcal{L}^*$ 
for any given variables $\bX, Y$. 
With this, one can propose the following decomposition based on the M\"obius
inversion formula~\cite{charalambides2018enumerative}:

\begin{definition}
For a given $p_{\bm X,Y}$, the atoms $S_\partial^{\bm\alpha}(\bX\rightarrow Y)$ 
correspond to the terms given by the \textit{M\"obius inverse} of
$S^{\bm\alpha}(\bX\rightarrow Y)$; i.e. the unique set of values that satisfy
\begin{equation}\label{eq:syn_atoms}
S_\partial^{\bm\alpha}(\bm X\rightarrow Y) \coloneqq S^{\bm\alpha}(\bm X\rightarrow Y) - \sum_{\subalign{\bm\beta&\in\mathcal{A}^*\\ \bm\beta&\succ\bm
\alpha }} S_\partial^{\bm\beta}(\bm X\rightarrow Y)
\vspace{-.3cm}
\end{equation}
for all $\bm\alpha\in\mathcal{A}^*$.
\end{definition}

Intuitively, the M\"obius inversion can be understood as a discrete derivative
over a lattice. In effect, an equivalent representation of the M\"obius
relationship is given by
\vspace{-0.1cm}
\begin{equation}\label{eq:fund_calculus}
S^{\bm\alpha}(\bm X\rightarrow Y) = \sum_{\subalign{\bm\beta&\in\mathcal{A}^*\\ \bm\beta&\succeq\bm
\alpha}} S_\partial^{\bm\beta}(\bm X\rightarrow Y)~,
\vspace{-0.1cm}
\end{equation}
which is analogous to the fundamental theorem of calculus. The M\"obius
inversion yields \emph{synergy atoms} of the form $S_\partial^{\bm\alpha}$,
which quantify how much information about the target is contained in the
collective effects of variables $\bm\alpha$. For example, $S^{[n]}(\bX\rightarrow Y) =
S_\partial^{[n]}(\bX\rightarrow Y) = 0$, and $S^{\emptyset}(\bX\rightarrow Y)
=I(\bX;Y)$ for any
$p_{\bm X,Y}$~\footnote{
Strictly speaking, since $\bm\alpha$ is a set of sets, these
atoms should be denoted by $S^{\{\emptyset\}}_\partial$,
$S^{\{\{1\}\{2\}\}}_\partial$, etc. For clarity, and for consistency with prior literature, we omit the outer bracket and denote these symbols by shortened expressions (e.g. $S^{\emptyset}_\partial$, $S^{\{1\}\{2\}}_\partial$, etc.).}. 
This last indentity, combined with Eq.~\eqref{eq:fund_calculus},
gives the following important result:

\begin{proposition}[Information decomposition]\label{prop:decomp}
The mutual information between $\bX$ and $Y$ can be decomposed as
\begin{align}\label{eq:decompppp}
  I(\bX; Y) = \sum_{\bm\alpha\in\mathcal{A}^*} S_\partial^{\bm\alpha}(\bX \rightarrow Y)~.
\end{align}
\end{proposition}
\begin{proof}
Follows directly from noting that $S^{\emptyset}(\bX\rightarrow Y) =I(\bX;Y)$,
and combining this with Eq.~\eqref{eq:fund_calculus}.
\end{proof}

The next section builds our intuition on this decomposition for small systems.

\subsection{The case $n=2$}\label{sec:decomposition_for_two}

\footnotetext[999]{Note that since $S^{\{12\}} = 0$, we have
$S_\partial^{\{1\}\{2\}} = S^{\{1\}\{2\}}$.}

After having formally presented the decomposition for $n$ variables, let us
focus on the bivariate ($n=2$) case, and develop some intuitions about the
resulting synergy atoms. For two predictors $\bm X=(X_1,X_2)$,
Equation~\eqref{eq:decompppp} yields
\begin{align*}
I(\bm X;Y) =& \: S_\partial^{\{1\}\{2\}}(\bX\rightarrow Y) + S_\partial^{\{1\}}(\bX\rightarrow Y) \\
                  & + S_\partial^{\{2\}}(\bX\rightarrow Y) + S_\partial^{\emptyset}(\bX\rightarrow Y) ~ .
\end{align*}

Above, $S_\partial^{\{1\}\{2\}}(\bX\rightarrow Y)$ can be understood as the information 
about $Y$ that is related to collective properties of $\bm X$ that can be disclosed 
without compromising either $X_1$ or $X_2$~\cite{Note999}. Similarly, $S_\partial^{\{1\}}(\bX\rightarrow Y)$ 
is the information about $Y$ that can be disclosed without revealing parts of $X_1$ but
compromising $X_2$ (otherwise it would have been included in $S_\partial^{\{1\}\{2\}}(\bX\rightarrow Y)$). 
Finally, $S_\partial^{\emptyset}(\bX\rightarrow Y)$ is information about $Y$ that compromises
both variables; put differently, information that is neither in $S^{\{1\}}(\bX\rightarrow Y)$
or $S^{\{2\}}(\bX\rightarrow Y)$. Loosely speaking,
$S_\partial^{\emptyset}$ can be associated with the standard PID redundancy,
$S_\partial^{\{i\}}$ with the unique information, and $S_\partial^{\{1\}\{2\}}$
with the synergy. A detailed comparison of these and the standard PID atoms is
presented in Section~\ref{sec:pid}.

\begin{figure}[t]
  \centering
  \includetikz{tikz/}{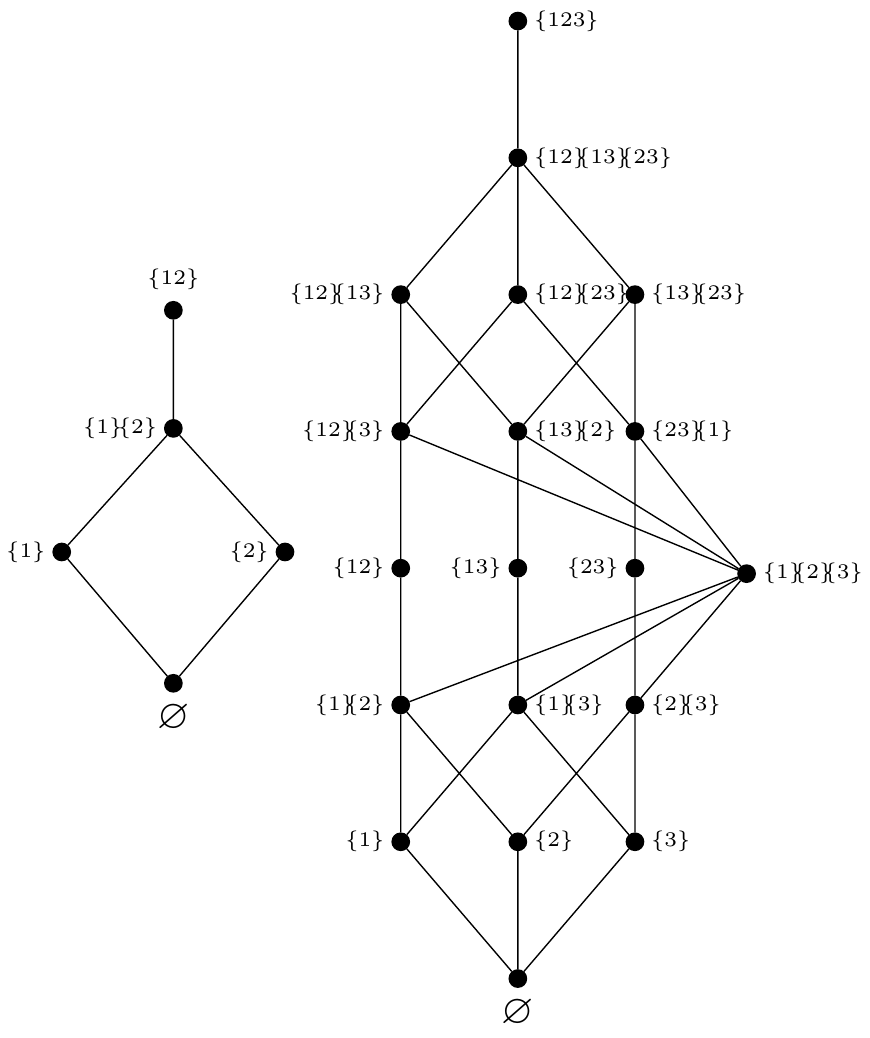}
  \caption{Extended constraint lattice for systems of $n=2$ (\emph{left}) and $n=3$ (\emph{right}) sources.}
  \label{fig:lattice}
\end{figure}

For the particular case where $X_1$ and $X_2$ are binary variables, then the
optimal synergistic channel only depends on their joint distribution -- and not
on the target variable, as shown in Ref.~\cite{rassouli2018latent}. Interestingly,
if $X_1$ and $X_2$ are independent fair coin flips, then~\footnote{Proofs of
this result can be found in Refs.~\cite{rassouli2018latent}
and~\cite{quax2017quantifying}.}
\begin{equation}
S^{\{1\}\{2\}}(\bX \rightarrow Y) = I( X_1\textnormal{\texttt{xor}} X_2 ;Y) ~ .
\end{equation}
This result shows that our definition of synergy effectively captures
high-order statistical effects, which are most purely exhibited by \texttt{XOR}
logic gates~\footnote{For more information on this relationship, please see
Refs.~\cite[Lemma~3]{rosas2019quantifying} and \cite[Section~VI]{rosas2018selforg}.}.
Analytical results for the more general case where $X_1$ and $X_2$ are binary, though not
necessarily independent, are presented in Appendix~\ref{app:bivariate}.

\begin{figure}[t]
  \centering
  \includetikz{tikz/}{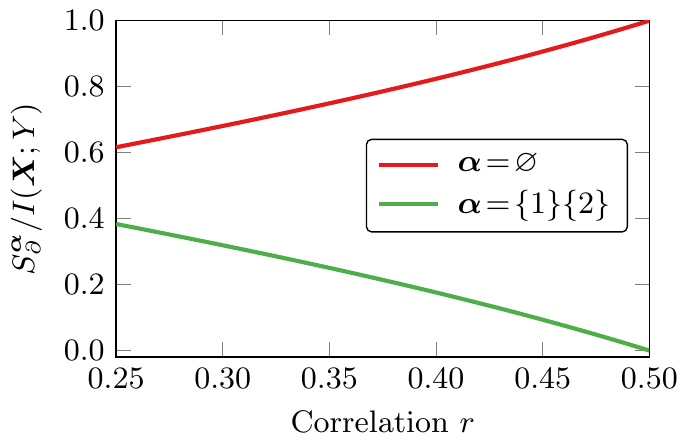}

  \caption{Normalised atoms of the disclosure decomposition for the
  \texttt{AND} gate with correlated inputs, with $\mathbb{P}\{
X_1=1\}=\mathbb{P}\{X_2=1\} = 0.5$ and correlation $\langle x_1 x_2 \rangle =
r$.}

  \label{fig:correlated_and}
\end{figure}

\begin{table}[h!]
  \centering
  \caption{Common distributions and their $S^{\bm\alpha}$ decomposition}
  \label{tab:dists}
  \ra{1.7}
  \setlength{\tabcolsep}{0.15cm}
  \begin{tabular}{lccccc}
    \toprule[0.35ex]
    ~ &        \texttt{XOR} & \texttt{COPY} & \texttt{Unq.1} & \texttt{AND} & \texttt{TBC} \\ 
    \midrule[0.2ex]                                                           
    \PI{\{1\}\{2\}} & 1     &       0        &       0       &    0.3113    &       1      \\        
    \PI{\{2\}}      & 0     &       0        &       1       &       0      &       0      \\ 
    \PI{\{1\}}      & 0     &       0        &       0       &       0      &       0      \\ 
    \PI{\emptyset}  & 0     &       1        &       0       &      0.5     &       1      \\ 
    \bottomrule[0.35ex]
  \end{tabular}
\end{table}

with these results, it is straightforward to compute the decomposition in
eq.~\eqref{eq:decompppp} for a few illustrative examples; results are presented
in table~\ref{tab:dists}. 
first, we notice that the paradigmatic distributions
\texttt{copy} and \texttt{xor} have the expected \SI{1}{\bit} of redundancy and
synergy, respectively, in agreement 
with our intuition for these cases. similarly, the \texttt{unq.1} distribution 
shows only one non-zero atom, $s_\partial^{\{2\}}$, which corresponds to unique information. 
the index of the atom, however, might seem counterintuitive; the confusion is explained
by the fact that the superscript $\{2\}$ refers to a constraint (the impossibility to
disclose what is in $x_2$), and hence $s_\partial^{\{2\}}$ is more related with the 
contents of $x_1$.
this shows a general theme: that $s^{\bm\alpha}$,
while operationally meaningful and intuitive, needs to be interpreted
differently from other pids (c.f. sec.~\ref{sec:pid}).

As a further example, we compute the disclosure decomposition
$S^{\bm\alpha}_\partial$ for the result of an \texttt{AND} gate with correlated
inputs (Fig.~\ref{fig:correlated_and}). As the inputs become more correlated,
there is less information that can be disclosed without compromising either of
them, and therefore the fraction of the total information that corresponds to
$S^{\emptyset}_\partial$ grows as correlation increases.

\section{The backbone decomposition}\label{sec:backbone}

As the extended constraint lattice $\mathcal{L}^*$ grows extremely rapidly with system
size, it is unfeasible to examine every element of our proposed 
decomposition in all but very small systems. Luckily, the nature of $S^{\bm\alpha}$ 
allows us to formulate a reduced collection of source-sets that form the ``backbone'' of the
constraint lattice, which provides a natural summary of the system's high-order
interactions. 

In the sequel, Subsection~\ref{seq:backbone_lattice} introduces the backbone lattice, then Subsection~\ref{seq:backbone_atoms} discusses the backbone decomposition, and finally Subsection~\ref{sec:examples} illustrates these ideas with some examples.

\subsection{The backbone constraint lattice}\label{seq:backbone_lattice}

We introduce the \emph{backbone constraint lattice}, denoted by $\mathcal{B}
\subset \mathcal{L}^*$, as the sublattice composed by the elements of $\mathcal{A}^*$ 
of the form $\bm\gamma_m = \{ \alpha \subset [n] : |\alpha| = m\}$ for
$m=0,\dots,n$ (the dependency on $n$ is left implicit). Importantly, 
$\Corder$ restricted to $\mathcal{B}$ provides a total order:
\begin{equation}
\bm\gamma_0 \Corder \bm\gamma_1 \dots \Corder \bm\gamma_n~.
\end{equation} 
For example, for the
case of $n=3$ then $\mathcal{B}$ is composed by 
$\bm\gamma_0=\{\emptyset\}$,
$\bm\gamma_1=\{\{1\},\{2\},\{3\}\}$, $\bm\gamma_2=\{\{1,2\},\{2,3\},\{1,3\}\}$,
and $\bm\gamma_3=\{\{1,2,3\}\}$. Hence, the constraint $\gamma_m$ corresponds
to the synergistic channel that discloses no information about any of the
$m$\textsuperscript{th}-order marginals.

For the synergy terms associated with $\mathcal{B}$, we use the shorthand notation $B^m(\bX\rightarrow
Y)\coloneqq S^{\bm\gamma_m}(\bX\rightarrow Y)$. In simple words, $B^m(\bX\rightarrow
Y)$ accounts for the information about $Y$ that can be disclosed
without compromising any group of $m$ variables. Furthermore, as $\bm\gamma_{m-1}
\Corder \bm\gamma_{m}$, the following chain of inequalities is guaranteed:
\begin{equation}\label{eq:diff}
0 = B^{n}(\bX \rightarrow Y) \leq \dots \leq B^0 (\bX \rightarrow Y) = I(\bX; Y).
\end{equation}

\subsection{Backbone atoms}\label{seq:backbone_atoms}

A new application of the M\"obius inversion formula allows us to
define \emph{backbone atoms}, $B_\partial^m(\bX\rightarrow Y)$, which we define
as
\begin{align}
B_\partial^m(\bX\rightarrow Y) 
&\coloneqq B^{m-1}(\bX\rightarrow Y) - \!\!\!\!\! \! \sum_{k=m+1}^n B_\partial^k(\bX\rightarrow Y) \nonumber\\
&= B^{m-1}(\bX\rightarrow Y) - B^{m}(\bX\rightarrow Y) \label{eq:diff_backbone} ~ .
\end{align}
Equivalently, the backbone atoms are the values $B_\partial^k(\bX\rightarrow Y)$ that satisfy, for all $m\in[n]$,
\begin{equation}\label{eq:backbone2}
B^{m-1}(\bX\rightarrow Y) = \sum_{k=m}^n B_\partial^k(\bX\rightarrow Y)~, 
\end{equation}
Intuitively, $B^{m-1}$ corresponds to the amount of information about
$Y$ that $\bX$ can reveal without compromising any group of $m-1$ variables; or, equivalently, information revealed by compromising only groups of $m$ or more variables. 
Consequently, $B_\partial^m$ quantifies the marginal gain of
information that can be disclosed by relaxing the constraints from groups of
$m$ variables to groups of $m-1$. For example, for $m=1$ then
$B^1(\bX\rightarrow Y)$ measures how much information can be disclosed while
keeping each $X_j$ confidential, while $B_\partial^1(\bX\rightarrow Y)$
corresponds to how much is gained when these constraints are relaxed.
Additionally, note that these backbone atoms can be directly related to the
synergy atoms in Eq.~\eqref{eq:syn_atoms}, as
\begin{align}\label{eq:backbone_as_PID}
B_\partial^m(\bX\rightarrow Y) = \sum_{\bm\gamma_{m-1} \Corder \bm\alpha \Corder \bm\gamma_m} S_\partial^{\bm\alpha}(\bX \rightarrow Y)~. 
\end{align}

Puting all these results together one finds a reduced decomposition, which is formalised
by the following result.
\begin{proposition}[Backbone decomposition]
The following decomposition always holds:
\begin{equation}\label{eq:decomposition2}
I(\bX\rightarrow Y) = \sum_{m=1}^{n} B_\partial^m(\bX\rightarrow Y)~.
\end{equation}
Moreover, $B_\partial^m(\bX\rightarrow Y) \geq 0$ for all $m=1,\dots,n$.
\end{proposition}
\begin{proof}
One can obtain Eq.~\eqref{eq:decomposition2} by evaluating Eq. \eqref{eq:backbone2} for $m=1$. 
The non-negativity of the atoms is a consequence of Eqs.~\eqref{eq:diff}
and~\eqref{eq:diff_backbone}.
\end{proof}

These backbone atoms provide a coarse-graining of the full decomposition in
Eq.~\eqref{eq:decompppp}. A basic schematic of this backbone decomposition, as
well as its relationship with the $S_\partial^{\bm\alpha}$ atoms in the
extended constraint lattice are shown in Fig.~\ref{fig:backbone_lat}.
Importantly, note that the cardinality of the backbone lattice grows linearly
with system size, and hence the number of atoms in
Eq.~\eqref{eq:decomposition2} remains tractable for large systems.

\subsection{Examples}
\label{sec:examples}

As an illustrative example of the potential of the backbone decomposition, let
us apply it to scenarios where the relationship between $\bX$ and $Y$ can be
expressed as a Gibbs distribution. In particular, we consider systems of $n+1$
spins (i.e. $\mathcal{X}_i = \{-1,1\}$ for $i=1,\dots,n+1$) whose joint
probability distributions can be expressed in the form
\begin{equation}\label{eq:gibbs}
p_{\bm X^{n+1}}(\bm x^{n+1}) = \frac{ e^{-\beta \mathcal{H}_k(\bm x^{n+1})}}{Z}~,
\end{equation}
where $\beta$ is the inverse temperature, $Z$ a normalisation constant, and
$\mathcal{H}_k(\bm x^n)$ a Hamiltonian function of the form
\begin{align}
\begin{split}
  \mathcal{H}_k(\bm x^{n+1}) = &- \sum_{i=1}^{n+1} J_i x_i - \sum_{i=1}^{n} \sum_{j=i+1}^{n+1} J_{i,j} x_i x_j \\ %
& \ldots - \sum_{|\bs{I}|=k} J_{\bs{\gamma}} \prod_{i\in \bs{I}} x_i~,
\end{split}
\label{eq:hamiltonian}
\end{align}
with the last sum running over all collections of indices $\bs{I}\subseteq
[n+1]$ of cardinality $|\bs{\gamma}| = k$. To calculate all quantities in this
section we consider $Y=X_{n+1}$ as target variable. Full simulation details are
reported in Appendix~\ref{app:simulation}.

\begin{figure}[t]
  \centering
  \includetikz{tikz/}{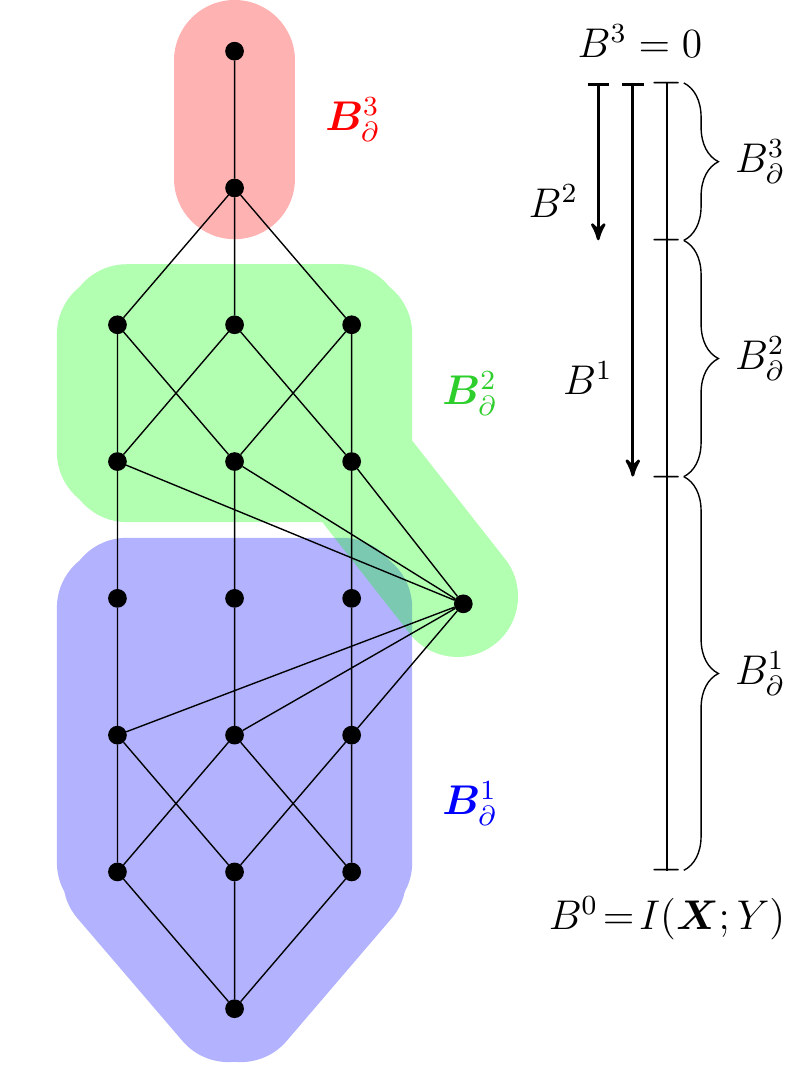}

  \caption{Schematic representation of the backbone lattice. (\emph{left})
  Correspondence between backbone atoms and $S_\partial^{\bm\alpha}$ for the
$n=3$ lattice. (\emph{right}) Representation of the backbone lattice as a
totally ordered set.}

  \label{fig:backbone_lat}
\end{figure}

\begin{figure}[h!]
  \centering
  \includetikz{tikz/}{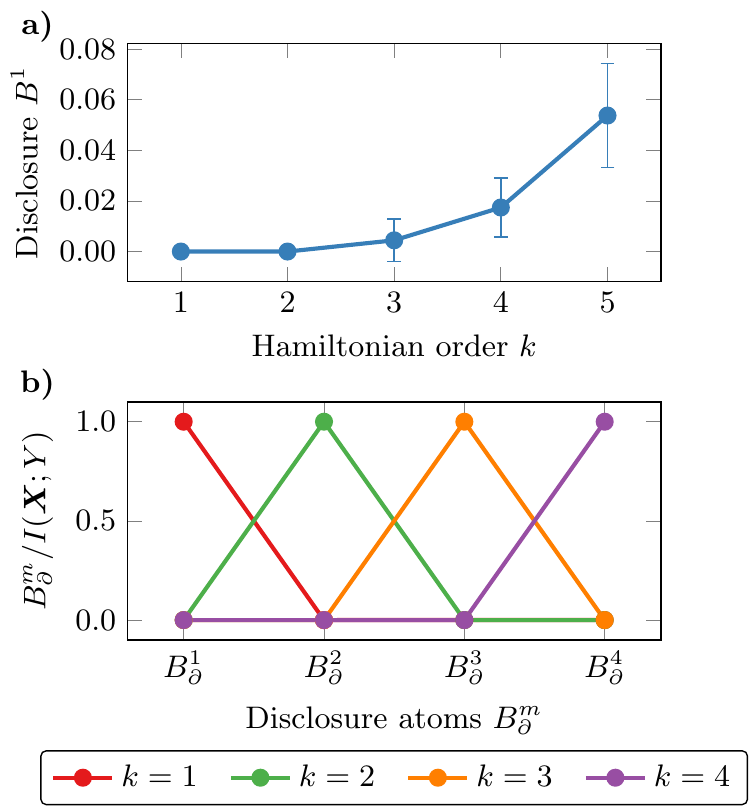}

  \caption{Synergistic disclosure in Ising models (\textbf{a})
  with terms up to order $k$ and (\textbf{b}) with terms
only of order $k$.}

  \label{fig:ising}
\end{figure}

As a first test case, we consider Hamiltonians with interactions up to order
$k$, as in Eq.~\eqref{eq:hamiltonian} above. For these systems, we calculated
the backbone term $B^1(\bX \rightarrow Y)$, which measures the strength of the
high-order statistical effects beyond pairwise interactions
(Fig.~\figsubref{fig:ising}{a}). As expected, our results show that if the
Hamiltonian only possesses first or second order interactions (i.e. $k=1$ or
$2$) then $B^1(\bX \rightarrow Y)$ is negligible; and for $k \geq 3$, $B^1(\bX
\rightarrow Y)$ grows monotonically with $k$.

As a second test case, we studied Hamiltonians with source-target interactions
only of order $k$, and compute their full backbone decomposition.
Fig.~\figsubref{fig:ising}{b} shows all the backbone atoms $B_\partial^m$ on
the X-axis, normalised by $I(\bX; Y)$. Interestingly, for each Hamiltonian
order $k$ there is only one non-zero backbone atom, which suggests that $I(\bX;
Y) \approx B_\partial^k(\bX \rightarrow Y)$. Note that this relationship
between Hamiltonian interaction order and backbone atom is highly non-trivial,
and finding analytical methods to make this connection more explicit is an open
question.

These findings suggest that the backbone decomposition may provide an analogue
to the measure of \emph{connected information} introduced in Refs.~\cite{amari2001information,schneidman2003network}, 
which captures the effects of Hamiltonian high-order terms over their corresponding Gibbs distributions~\footnote{Note that
the presented findings consider a particular ensemble of distributions. Therefore, more work is needed in order to explore their generality.}. 
The main difference between the connected information and the backbone decomposition is that in the former
all variables play an equivalent role, while in the latter they are divided between sources and target.

\section{Synergistic capacity and private self-disclosure}
\label{sec:selfsyn}

So far, we have investigated the usual information decomposition scenario, in
which a group of source variables \bX hold information about \emph{another},
target variable $Y$. Using the tools developed so far, we can ask a new
question: how much information can \bX disclose \emph{about itself} under
specific constraints? Answering this question will provide further
intuitions on the nature of synergistic disclosure, while revealing some
unexpected properties.

We start by presenting the definition of the \emph{self-disclosure} of a
system, which is a particular case of the formalism presented above.

\begin{definition}
The $\bm\alpha$-self-synergy of \bX is given by $S^{\bm\alpha}(\bX \rightarrow
\bX)$, and denoted simply by $S^{\bm\alpha}(\bX)$.
\end{definition}

This definition makes it straightforward to extend the concepts above to define
self-synergy atoms $S_\partial^{\bm\alpha}(\bX)$, as well as backbone
self-synergy terms and atoms, denoted by $B^{m}(\bX)$ and
$B_\partial^{m}(\bX)$, respectively.

Let us begin with an example, by computing the self-disclosure of binary
bivariate distributions. Consider two binary variables $\bX = (X_1,X_2)$, with
$\mathbb{P}\{ X_1=1\}=\mathbb{P}\{X_2=1\} = p$ and
$\mathbb{P}\{X_1=1,X_2=1\}=r$ (Fig.~\ref{fig:selfdisc}). Perhaps surprisingly,
a direct calculation shows that maximal synergy is achieved for $X_1,X_2$
independent and $p=1/2$ -- which is equivalent to the much-debated Two-Bit-Copy
(TBC) gate commonly discussed in the PID
literature~\cite{griffith2014quantifying,finn2018pointwise,harder2013bivariate}.
To make sense of this result, consider the following bounds on the
self-disclosure:

\begin{lemma}\label{lemma:self_syn_upper_bound}
For any $\bX,Y$ the following bound holds:
\begin{equation}
 H(\bX) - \max_{\alpha_j \in \bm\alpha} H(\bm X^{\alpha_j})  \geq S^{\bm\alpha}(\bX) \geq S^{\bm\alpha}(\bX \rightarrow Y)~.
\end{equation}
\end{lemma}
\begin{proof}
  The upper bound is proven by an application of Proposition~\ref{pr:upper_bound} with $Y = \bX$, and the lower bound by an application of Lemma~\ref{lemma:DPI}. 
\end{proof}

This lower bound is particularly insightful, as it
suggests that the synergistic self-disclosure of $\bX$ is the \emph{tightest
upper bound on the synergistic information that $\bX$ could hold about any
other target}. Therefore, this (admittedly heterodox) perspective of
synergy provides a clear explanation of why the TBC could have non-zero synergy, 
since it accounts for the ``synergistic capacity'' of its inputs.

\begin{figure}[t]
  \centering
  \includetikz{tikz/}{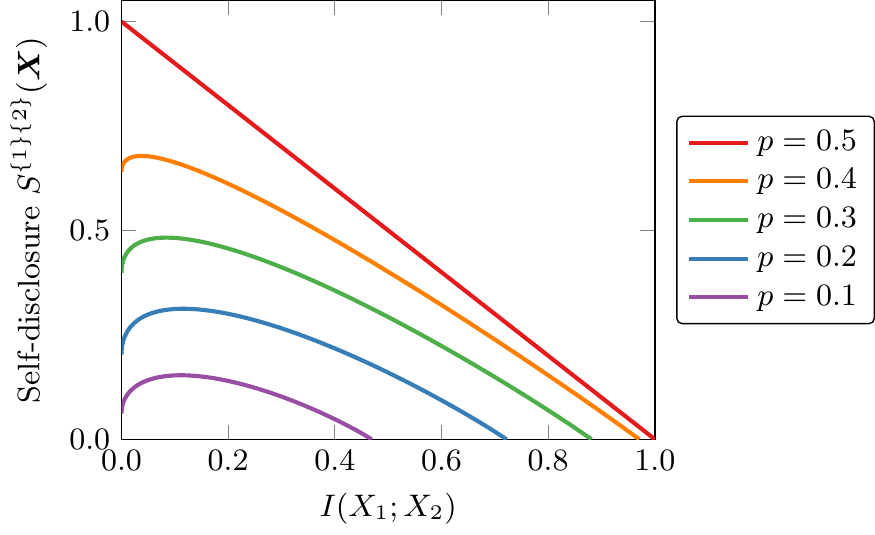}
  \caption{Self-disclosure capacity of two correlated bits.}
  \label{fig:selfdisc}
\end{figure}

Additionally, the upper bound in Lemma~\ref{lemma:self_syn_upper_bound} provides a quick way
to estimate how much synergy can be found with respect to a given set of sources $\bX$. For example,
if $(X_1,X_2)$ are two i.i.d. fair coins, Lemma~\ref{lemma:self_syn_upper_bound} states that their synergy
cannot be larger than 1 bit, which is attained by the optimal self-synergistic channel 
$V^* = X_1 \texttt{xor} X_2$ \footnote{Since our framework provides a algorithm to build the optimal self-synergistic channel for arbitrary sources $\bX$, it would be natural to conjecture that this channel could also be optimal for other target
variables -- i.e., that it could serve as sufficient statistic under the corresponding
constraints. Unfortunately, numerical explorations show that, while this works
for two binary variables (Appendix~\ref{app:bivariate}), it is in general not the case.}.

Another natural conjecture, in the light of the findings reported in
Section~\ref{sec:examples}, would be to argue a relationship between
self-synergy and connected information, as both measures treat
symmetrically all the corresponding variables. However, numerical evaluations
show there is no relationship between them. As a matter of fact, systems with
low degrees of interdependency have high levels of self-synergy, while having
low levels of connected information.

A final lesson that can be learnt from studying self-synergy is that high-order
synergies are not rare corner cases, but are in fact prevalent in the space of
probability distributions. More formally, our next result shows that $B^m(\bX)$
takes most of the information contained in $\bX$ as the system size grows.

\begin{proposition}\label{prop:asynto}
Consider a sequence of random variables $\bX \coloneqq (X_1,\ldots,X_n)$ for which there exists $K\in\mathbb{N}$ such that 
$|\mathcal{X}_k| \leq K$ for all $k\in\mathbb{N}$. If
$\lim_{n\rightarrow\infty} H(\bX)/n$ exists and is not zero, then for any fixed $m\in\mathbb{N}$
\begin{equation}
\lim_{n\rightarrow\infty} \frac{ B^{m}(\bX)}{H(\bX)} = 1 ~ .
\end{equation}
\end{proposition}
\begin{proof}
See Appendix~\ref{app:dominance}.
\end{proof}
Let us work an example to gain intuition on this seemingly counterintuitive
result.
\begin{example}
Consider a system $\bX$ where the components $X_k$ are independent
fair coins. The
mapping $V= (X_1\textnormal{\texttt{xor}} X_2, \dots,
X_{n-1}\textnormal{\texttt{xor}} X_n): \{0,1\}^n\to\{0,1\}^{n-1}$ belongs to $C(\bX, \{\{1\},\dots,\{n\}\})$, and $I(V;\bX) =
n-1$, attaining the upper bound provided in Lemma~\ref{lemma:self_syn_upper_bound}. This implies that $B^1(\bX) =
n-1$. Similarly, one can notice that $V_m = (X_1\textnormal{\texttt{xor}} \dots \textnormal{\texttt{xor}}
X_m, \dots, X_{n-m+1}\textnormal{\texttt{xor}}\dots \textnormal{\texttt{xor}}
X_n): \{0,1\}^n\to\{0,1\}^{n-m+1}$ also attains the bound for the class $C(\bX,\{\{1,\dots,m\},\dots,\{n-m+1,\dots,n\}\})$, and hence
$B^m(\bX) = n-m$. Therefore, $B_\partial^m(\bX) = 1$ for all $m=1,\dots,n-1$, and
\begin{equation}
\lim_{n\rightarrow\infty} \frac{ B^{m}(\bX)}{H(\bX)} = \lim_{n\rightarrow\infty} \frac{ n-m }{n} = 1~.
\end{equation}
\end{example}

The theoretical and practical consequences of the prevalence of synergy will be
discussed in a separate publication.

\section{Relationship with other information decompositions}
\label{sec:pid}

This section explores the relationship of our proposed framework with other information
decompositions. For this, Subsection~\ref{sec:axioms} explores various properties of
our definition of synergy under the light of various axioms typically used in the PID literature,
then Subsection~\ref{sec:relationship_pid} explores relationships of our decomposition with
other PID, and finally Subsection~\ref{sec:numerical} carries out numerical comparisons between
our metrics and other well-known decompositions.

\subsection{Axioms}
\label{sec:axioms}

In previous literature, partial information decomposition is usually discussed
in terms of axioms, which encode various desirable properties that measures
might -- or might not -- satisfy. These axioms are often formulated for
redundancy measures, which, given that the basic constituent of our
decomposition is a synergy measure, makes assessing our framework in these
terms non-trivial. Nevertheless, this subsection explores some of the common
axioms from the point of view of $S^{\bm\alpha}$, using as guideline the set of
axioms discussed in Ref.~\cite{griffith2014intersection}.

The following axioms are satisfied by our measure:
\begin{itemize}
\item \textbf{(GP)} Global positivity: $S^{\bm\alpha}(\bX\rightarrow Y) \geq 0$ for all $\bX,Y$ and $\bm\alpha\in\mathcal{A}^*$.

\item \textbf{(Eq)} Equivalence-class invariance: $S^{\bm\alpha}(\bX\rightarrow Y)$ is invariant under substitution of $X_i$ or $Y$ by an informationally equivalent random variable (i.e. re-labeling).

\item \textbf{(wS)} Weak symmetry: $S^{\bm\alpha}(\bX\rightarrow Y)$ is invariant under reordering of $X_1,\dots,X_n$.

\item \textbf{(wM)} Weak monotonicity: $S^{\bm\alpha}(\bX\rightarrow Y) \leq S^{\bm\alpha}\big((\bX,Z)\rightarrow Y\big)$ (see Appendix~\ref{app:Corder}). Note that this doesn't hold for the backbone terms, as the $\bm\alpha$'s are not equal.

\item \textbf{(CCx)} Channel convexity: 
$S^\mathbf{\alpha}(\mathbf{X}\rightarrow Y)$ is a convex function of $p_{Y|\mathbf X}$ for a given $p_{\bf X}$ (proof in Appendix~\ref{app:further_props}).

\item \textbf{(T-DPI)} Target data processing inequality: if $\bX - Y - Z$ is a Markov chain, then $S^{\bm\alpha}(\bX \rightarrow Y) \geq S^{\bm\alpha}(\bX \rightarrow Z)$ (proof in Appendix~\ref {app:further_props}).

\end{itemize}

The proposed measure does not satisfy strong symmetry (\textbf{sS}), as it
might be the case that $S^{\bm\alpha}\big((X_1,X_2)\rightarrow Y\big) \neq
S^{\bm\alpha}\big((X_1,Y)\rightarrow X_2\big)$~\footnote{As an example of this,
if \unexpanded{$(X_1,X_2)$} are two independent fair coins and
\unexpanded{$Y=(X_1,X_2)$}, then a direct calculation shows that, if
\unexpanded{$\bm\alpha=\{\{1\},\{2\}\}$}, then
\unexpanded{$S^{\bm\alpha}\big((X_1,X_2)\rightarrow Y\big) = 1$} and
\unexpanded{$S^{\bm\alpha}\big((X_1,Y)\rightarrow X_2\big) = 0$}.}.

We can prove by counterexample that $S_\partial^{\bm\alpha}$ does not satisfy
strong local positivity (\textbf{LP}), i.e. that there exist
$S^{\bm\alpha}_\partial(\bX\rightarrow Y) < 0$ for some
$\bm\alpha\in\mathcal{A}^*$~\footnote{Consider a
``double-\texttt{XOR}'' distribution, with 3 independent bits as inputs, and
\unexpanded{$Y = (X_1 \texttt{xor} X_2, X_2 \texttt{xor} X_3)$} as output. For
this distribution, all atoms \unexpanded{$S_\partial^{\{ij\}\{k\}}(\bX
\rightarrow Y) = -1$}, violating (\textbf{LP}). To see why, note that
\unexpanded{$S^{\{12\}\{13\}}(\bX \rightarrow Y) = S^{\{12\}\{3\}}(\bX
\rightarrow Y) = 1$}, since in both cases \unexpanded{$X_2 \texttt{xor} X_3$}
can be disclosed, yielding the second bit of \unexpanded{$Y$}; and
\unexpanded{$S^{\{12\}\{23\}}(\bX \rightarrow Y) = 1$}, since \unexpanded{$X_1
\texttt{xor} X_3$} can be disclosed, yielding the parity of \unexpanded{$Y$}.
Hence, applying the M{\"o}bius inversion, we have
\unexpanded{$S_\partial^{\{12\}\{3\}}(\bX \rightarrow Y) = -1$}.}.
On the other hand, 
note that the backbone atoms $B^m_\partial(\bX\rightarrow Y)$ do satisfy
(\textbf{LP}), as shown in Section~\ref{sec:backbone}.

\subsection{General relationship with PID}
\label{sec:relationship_pid}

In this section we focus on the relationship between our decomposition for the
case of $n=2$ (c.f. Sec.~\ref{sec:decomposition_for_two}), and the standard PID. 
When considering $\bm\alpha, \bm\beta \in \mathcal{A}\coloneqq \mathcal{A}^* / \{\emptyset\}$, the 
classic work of Williams and Beer~\cite{williams2010nonnegative} introduces
the following partial ordering:
\begin{equation}
\bm\alpha \preceq_\text{wb} \bm \beta \iff \forall \beta \in \bm\beta ~ \exists \alpha \in \bm\alpha, \alpha \subseteq \beta~. 
\end{equation}
While the difference between $\preceq_\text{wb}$ and $\preceq_\text{c}$ might
seem subtle, they induce drastically different lattice structures. For example,
if $\bm\alpha=\{\{1\}\}$ and $\bm\beta=\{\{1\}\{2\}\}$, then $\bm\beta
\preceq_\text{wb} \bm\alpha$ while $\bm\alpha \preceq_\text{c} \bm\beta$. The
lattices for $n=2$ for both orderings are shown in
Fig.~\ref{fig:wb_cons_lattices}.

Traditional PID-type decompositions for two sources are based on the following
conditions:
\begin{align*}
I(X_i;Y) &= \texttt{Red}(X_1,X_2\rightarrow Y) + \texttt{Un}(X_i; Y | X_j) \\ 
I(X_i;Y|X_j) &= \texttt{Un}(X_i; Y | X_j)  + \texttt{Syn}(X_1,X_2\rightarrow Y) ~, 
\end{align*}
which are valid for $i,j\in\{1,2\}$ with $i\neq j$. A direct parallel between
these terms and our framework can be made, and is shown in
Table~\ref{tab:correspondence}.

\begin{figure}[ht]
  \centering
  \includetikz{tikz/}{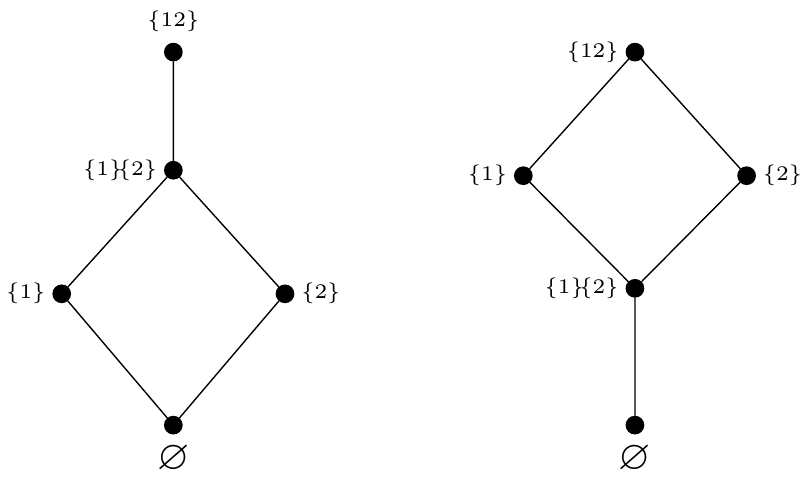}
  \caption{Extended constraint (\emph{left}) and redundancy (\emph{right})
  lattices for $n=2$.}
  \label{fig:wb_cons_lattices}
\end{figure}

A key relationship between any PID and our decomposition comes from noticing that,
considering Proposition~\ref{pr:upper_bound} for $\bm X=(X_1,X_2)$ and
$\bm\alpha=\{\{1\}\}$, one finds that
\begin{equation}
S^{\{1\}}(\bm X \rightarrow Y) \leq I(X_2;Y|X_1)~.
\end{equation}
Moreover, numerical evaluations show that this bound is often not attained, as
illustrated by Figure~\ref{fig:CMI} (see Appendix~\ref{app:simulation} for more
details).

\begin{figure}[t]
  \centering
  \includetikz{tikz/}{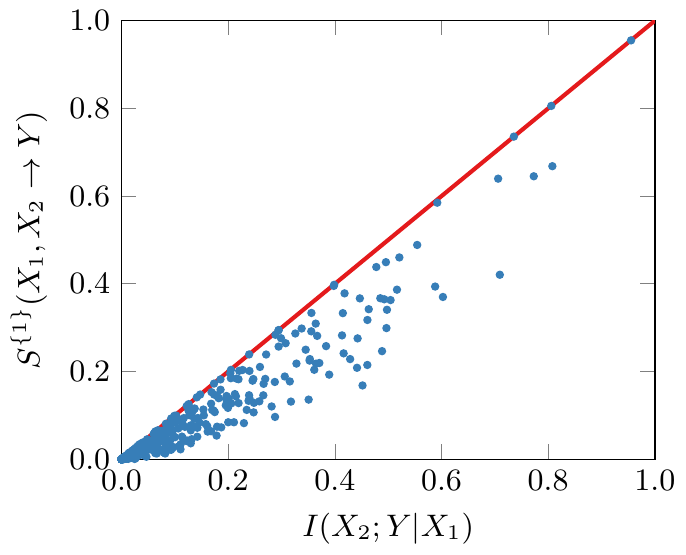}
  \caption{Conditional information is an upper bound on discloseable information.}
  \label{fig:CMI}
\end{figure}

As a consequence of this, one has that
\begin{align}
I(X_i;Y|X_j) &= \texttt{Un}(X_i; Y | X_j)  + \texttt{Syn}(X_1,X_2\rightarrow Y)  \nonumber \\
&\geq S^{\{j\}}(X_1,X_2\rightarrow Y)  \nonumber\\
&= S_\partial^{\{j\}}(X_1,X_2\rightarrow Y) + S_\partial^{\{1\}\{2\}}(X_1,X_2\rightarrow Y) ~. \nonumber
\end{align}
Conversely, an opposite relationship holds for the marginal mutual information:
\begin{align}
I(X_i;Y) &=\texttt{Red}(X_1,X_2\rightarrow Y) + \texttt{Un}(X_i; Y | X_j)~. \nonumber  \\
&\leq I(X_1,X_2;Y) - S^{\{i\}}(X_1,X_2\rightarrow Y) \nonumber\\
&= S_\partial^{\emptyset}(X_1,X_2\rightarrow Y) + S_\partial^{\{i\}}(X_1,X_2\rightarrow Y)~.\nonumber
\end{align}
By combining these two results, one can compare the co-information with a corresponding
co-information obtained from our decomposition, as follows:
\begin{align}
I(X_1;X_2;Y) &=  \texttt{Red}(X_1,X_2\rightarrow Y) - \texttt{Syn}(X_1,X_2\rightarrow Y) \nonumber \\
&= I(X_i;Y) - I(X_i;Y|X_j) \nonumber \\
& \geq S_\partial^{\emptyset}(\bX \rightarrow Y) - S_\partial^{\{1\}\{2\}}(\bX \rightarrow Y)~.\nonumber
\end{align} 

This result implies that, when assessing the balance between redundancy
and synergy, our decomposition always tends towards redundancy over synergy with
respect to any PID decomposition. In this sense, one can say that -- at least
for $n=2$ -- our decomposition is conservative when attributing dominance of
synergies. The next section provides further evidence to support this claim.

\subsection{Numerical comparisons with other PIDs}
\label{sec:numerical}

Let us now study how our proposed measure of synergy relates to the ones
corresponding to other well-known decompositions. 
Our analysis includes the $I_{\mathrm{BROJA}}$ decomposition by Bertschinger \emph{et
al.}~\cite{bertschinger2014quantifying}, \textit{Common Change in Surprisal}
($I_{\mathrm{ccs}}$) by Ince~\cite{ince2017measuring}, $I_{\mathrm{min}}$ by Williams and
Beer~\cite{williams2010nonnegative}, $I_{\mathrm{dep}}$ by James \emph{et
al.}~\cite{james2018unique}, and the pointwise decomposition by Finn and
Lizier ($I_\pm$)~\cite{finn2018pointwise}; all computed using the \texttt{dit} package~\cite{james2018dit}. To do this comparison, we draw random
distributions from the probability simplex following a NSB prior (see
Appendix~\ref{app:simulation} for details), and then compute their synergy
values with all measures.

\begin{table}[t]
  \centering
  \caption{Correspondence between PID atoms and $S_\partial^{\bm\alpha}$}
  \label{tab:correspondence}
  \ra{1.7}
  \setlength{\tabcolsep}{0.15cm}
  \begin{tabular}{lcc}
    \toprule[0.35ex]
    ~ &       Disclosure decomposition  & PID \\
    \midrule[0.2ex]                                                           %
 \textit{Synergy} & $S^{\{1\}\{2\}}(\bX \rightarrow Y)$ & $\texttt{Syn}(X_1,X_2 \rightarrow Y)$ \\ 
 \textit{Unique} & $S_\partial^{\{i\}}(\bX \rightarrow Y)$       & $\texttt{Un}(X_j;Y|X_i)$ \\
 \textit{Redundancy} & $S_\partial^{\emptyset}(\bX \rightarrow Y)$ & $\texttt{Red}(X_1,X_2\rightarrow Y)$ \\
    \bottomrule[0.35ex]
  \end{tabular}
\end{table}

\begin{figure}[ht]
  \centering
  \includetikz{tikz/}{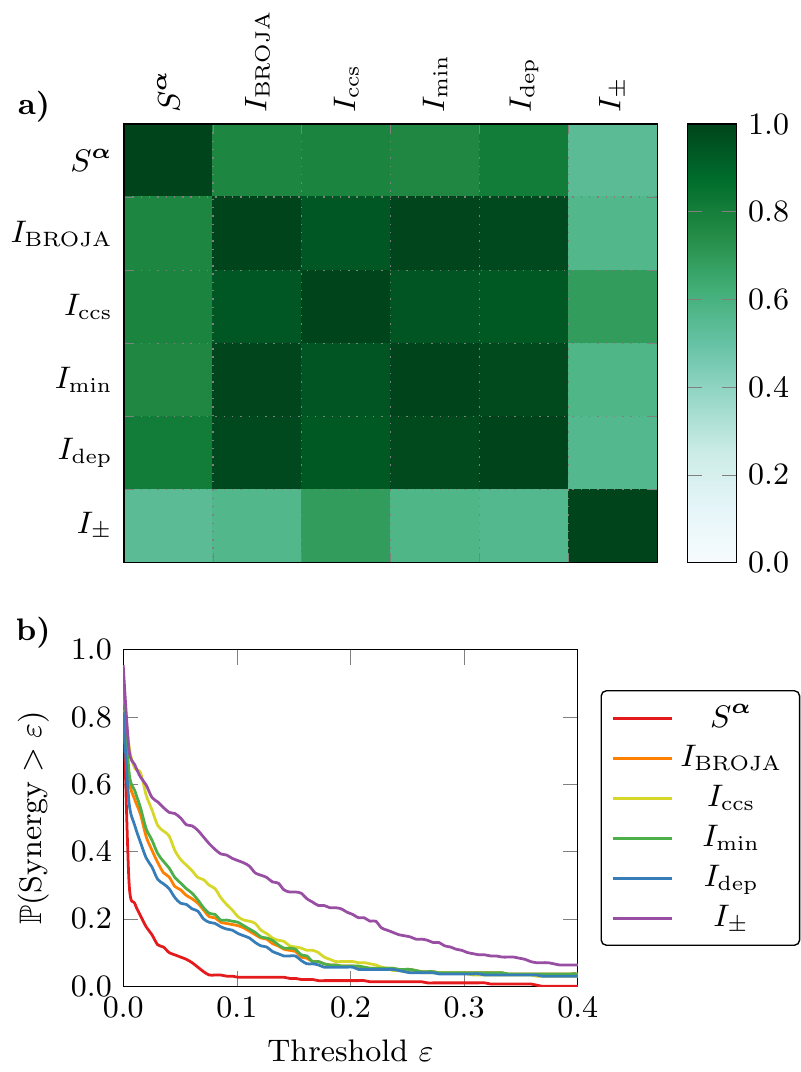}
  \caption{Numerical comparison between synergy values according to
  different proposed decompositions. (\textbf{a}) Correlation matrix
  of synergy values of random distributions. (\textbf{b}) Fraction of
  distributions with synergy greater than a given threshold.}
  \label{fig:synergycomp}
\end{figure}

A first, somewhat striking result is the overwhelming correlation found between
most proposed measures -- BROJA, CCS, $I_{\mathrm{min}}$, and $I_{\mathrm{dep}}$ are all related with each
other with correlations greater than 0.94 for every pair
(Fig.~\figsubref{fig:synergycomp}{a}). The two oddballs in this plot are our proposed
measure $S^{\bm\alpha}$
and $I_{\pm}$, which are less well correlated with
the rest \emph{and} with each other (correlations range around 0.70 for
$S^{\gamma_1}$ and around 0.50 for $I_{\pm}$).

To examine this discrepancy, we computed the inverse cumulative function of the
resulting values of the synergy for the various measures
(Fig.~\figsubref{fig:synergycomp}{b}). This curve shows the fraction of all
sampled distributions that have a synergy greater than a given threshold, to
gauge how prevalent synergy is judged to be according to each measure.
Consistent with Fig.~\figsubref{fig:synergycomp}{a},
Fig.~\figsubref{fig:synergycomp}{b} shows that the measures BROJA, CCS, $I_{\mathrm{min}}$, and
$I_{\mathrm{dep}}$ all follow similar profiles. Interestingly, $S^{\gamma_1}$ falls
much faster than the rest, while $I_{\pm}$ does it much more slowly.
Therefore, our measure $S^{\gamma_1}$ can be said to be more ``restrictive,''
in the sense that it tends to assign lower vaules of synergy, while
$I_{\pm}$ is more lenient. We hypothesise this ``overestimation'' of synergy by $I_\pm$
happens because of its tendency to assign negative values to the redundant or
unique information~\footnote{We do not take a stance here with respect to the
non-negativity of information atoms; but since the atoms have to sum to the
same mutual information, negative values in the lower atoms necessarily entail
inflated synergy values.}.

\section{Conclusion}
\label{sec:conclusion}

This paper puts forward an operational definition of informational synergy, and uses it as a foundation to build a multivariate information decomposition. Compared to previous approaches to information decomposition, our framework possesses two key features:

\begin{enumerate}

  \item It is a ``synergy-first'' decomposition, which begins by positing a
  measure of synergy and builds a decomposition after it, as opposed to
  previous approaches that are based on redundancy or unique information, and
  have synergy as a by-product.

  \item It is based on a quantity that is the optimal solution of a well-defined
  problem in the data privacy literature, which makes reasoning about
  the measure more transparent while bringing the decomposition closer to standard
  information-theoretic formulations.

\end{enumerate}

We illustrated the capabilities of the proposed decomposition on various examples, 
and showed that it gives a complementary perspective compared to other information 
decompositions. In particular, our results show that our measure of synergy is in general 
more conservative than other approaches, as it tends to attribute smaller values of synergy. 
We also showed how its operational interpretation provides clear explanations
to open questions in the field of information decomposition, such as the
well-known two-bit-copy problem~\cite{ince2017measuring,finn2018pointwise,kolchinsky2019novel}.

Moreover, our measure has an associated ``backbone'' decomposition, which provides a 
natural coarse-graining of the information atoms. Our results show that in some scenarios the backbone 
atoms provide a directed version of the well-known connected information, which captures
the effect of high-order interaction terms within Gibbs distributions. 
The number of backbone atoms grows linearly with system size, which 
makes this decomposition practical for studying a wide range of systems of interest.

The operational approach taken in this work represents a step towards establishing a solid foundation in the field of
information decomposition. Additionally, we provide an open-source software package~\footnote{A Python implementation of synergistic disclosure and the corresponding decomposition can be found online at \texttt{\url{https://github.com/pmediano/syndisc}}.} implementing the key quantities in this paper,
opening the door for a wide range of applications in data analysis, neuroscience, and information dynamics.

\section*{Acknowledgements}

The authors thank Michael Gastpar and Shamil Chandaria for insightful discussions, and Yike Guo for supporting this research. F.R. is supported by the Ad Astra Chandaria foundation. P.M. is funded by the Wellcome Trust (grant no. 210920/Z/18/Z).


\appendix

\section{Characterising synergistic channels}
\label{app:channel_description}

Here we provide a characterisation of $\mathcal{C}(\bX; \bm\alpha)$ in terms of matricial properties of it constituents. To do this, let us introduce the matrix $\mathbf{P}_{\bm\alpha}$ defined as
\begin{equation}\label{eq:matp}
\mathbf{P}_{\bm\alpha}\triangleq\begin{bmatrix}
\mathbf{P}_{\bm X^{\alpha_1}| \bm X}\\ \vdots \\\mathbf{P}_{\bm X^{\alpha_L}| \bm X}
\end{bmatrix}_{G\times|\hat{\mathcal{X}}|},
\end{equation}
where $G\coloneqq \sum_{k=1}^L \prod_{i\in \alpha_k}|\mathcal{X}_i|$, and
$\hat{\mathcal{X}}$ is the set of tuples $\bm x\in \prod_{k=1}^n \mathcal{X}_k$ such
that $p_{\bX}(\bm x) >0$.  
This matrix is designed such that the matrix product $\mathbf{P}_{\bm\alpha}\mathbf{p}_{\bX}$ 
(with $\mathbf{p}_{\bX}$ being the probability vector of $\bX$) 
yields the marginals within $p_{\bX}$ that need to be ``masked'' by the synergistic
channel -- so that $p_{\bX^{\alpha}|V=v}$ is a uniform distribution for all $\alpha\in\bm\alpha$. 
Note that $\mathbf{P}_{\bm\alpha}$ is a binary matrix, since the $\bm X^{\bm\alpha}$'s are
deterministic functions of $\bm X$. As an example, if $|\mathcal{X}_i|=2, \forall i\in[n]$ and
$\bm\alpha=\{\{1\},\dots,\{n\}\}$, then $\mathbf{P}_{\bm\alpha}$ is a $2n\times
2^n$ matrix that can be built recursively according to
\begin{equation}\label{matrix_structure}
\mathbf{P}_{k+1} = 
    \left[
    \begin{array}{cc}
    1 \dots 1 & 0 \dots 0 \\
    0 \dots 0 & 1 \dots 1 \\
    \mathbf{P}_k & \mathbf{P}_k
    \end{array}
    \right],  
    \nonumber
\end{equation}
with $\mathbf{P}_{\bm\alpha} = \mathbf{P}_n$ and $\mathbf{P}_1 = \Big[
\begin{array}{cc} 1 0 \\ 0 1\end{array}\Big] $. With this matrix, one can characterise the channels in
$\mathcal{C}(\bX; \bm\alpha)$ as shown in the next lemma, this being a straightforward extension of Ref.~\cite[Lemma~1]{rassouli2018latent}.

\begin{lemma}\label{lemma:markov}
$p_{V|\bm X} \in \mathcal{C}(\bm X; \bm\alpha)$ if and only if 
$(\mathbf{p}_{\bX}-\mathbf{p}_{\bX|v})\in\textnormal{Null}(\mathbf{P}_{\bm\alpha}),\forall v\in\mathcal{V}$.
\end{lemma}
\begin{proof}
\label{app:lemma_markov}

Let $X$, $Y$ and $Z$ be discrete r.v.'s that form a Markov chain as $X-Y-Z$.
Having $X\independent Z$ is equivalent to $ p_{X}(\cdot)=p_{X|Z}(\cdot|z)$,
i.e., $\mathbf{p}_{X}=\mathbf{p}_{X|z},\ \forall z\in\mathcal{Z}$. Furthermore,
due to the Markov chain assumption, we have $\mathbf{p}_{X|z} =
\mathbf{P}_{X|Y}\mathbf{p}_{Y|z},\ \forall z\in\mathcal{Z}$, and in particular,
$\mathbf{p}_{X}=\mathbf{P}_{X|Y}\mathbf{p}_{Y}$. Therefore, having
$\mathbf{p}_{X}=\mathbf{p}_{X|z}$, $\forall z\in\mathcal{Z}$ results in
\begin{equation*}
\mathbf{P}_{X|Y} \left(\mathbf{p}_{Y}- \mathbf{p}_{Y|z} \right)=\mathbf{0},\:\forall z \in\mathcal{Z},
\end{equation*}
or equivalently, $\left(\mathbf{p}_{Y}-\mathbf{p}_{Y|z}\right)\in
\text{Null}(\mathbf{P}_{X|Y})$, $\forall z\in\mathcal{Z}$.

The proof is complete by noting that i) $\bm X^{\alpha_i}- \bm X-Y$ form a
Markov chain for each index $i\in[L]$, and ii)
$\mbox{Null}(\mathbf{P}_{\bm\alpha}) = \bigcap_{i=1}^n \mbox{Null}(\mathbf{P}_{\bm
X^{\alpha_i}|\bm X})$.

\end{proof}

In summary, the matrix form of the (reverse) synergistic channel is related
to $p_{\bX}$ and the null space of $\mathbf{P}_{\bm\alpha}$. They key take-away
from this lemma is that one can compute the conditional distributions
$p_{\bX|v}$ of the synergistic channel by algebraic manipulation of
$\mathbf{P}_{\bm\alpha}$ and $\mathbf{p}_{\bX}$. Furthermore, this lemma has one
important implication: the synergistic channels needed to compute the
synergistic components of $I(\bX; Y)$ with respect to a target variable $Y$ 
depend only on $p_{\bX}$, not on $p_{Y|\bX}$.

\section{Synergy based on $f$-information}
\label{app:finfo}

Let $p$,$q$ be two probability mass functions on $\mathcal{X}$, such that $q(x)>0, \forall x\in\mathcal{X}$. For a convex function $f$ such that $f(1)=0$, the $f$-divergence of $p$ from $q$ is defined as
\begin{equation}\label{fdivergence}
    D_f(p||q)\coloneqq\sum_{x\in\mathcal{X}}f\left(\frac{p(x)}{q(x)}\right)q(x).
\end{equation}
Many well-known divergences are special cases of the $f$-divergence, including the Kullback-Leibler divergence (for $f(p)=p\log p$) and the total variation distance (for $f(p)=|p-1|/2$). 

Using $D_f$, one can define the \textit{f-information} of a pair of discrete random variables $(X,Y)$ as
\begin{equation}\label{finformation}
    I_f(X;Y)\coloneqq D_f(p_{X,Y}||p_X\cdot p_Y),
\end{equation}
which assesses the difference in terms of the $f$-divergence between their joint pmf (i.e. $p_{X,Y}$) and the product of the marginals (i.e., $p_X\cdot p_Y$). As a special case, one can obtain Shannon's mutual information when $f(p)=p\log p$.

Since the main tools used in this paper (such as convexity and data processing inequality) are also valid for the $f$-information, the main results of this paper continue to hold for the more general \textit{$f$-synergy}, defined as
\begin{equation}
S_f^{\bm \alpha}(\bX \rightarrow Y)\coloneqq \sup_{\substack{p_{V|\bm X}\in \mathcal{C}(\bm X;\bm\alpha):\\V-\bm X-Y}} I_f(V;Y)~.
\end{equation} 
That being said, please note that the $f$-information in general does not satisfy a chain rule, unlike mutual information which stems from the logarithmic nature of the KL-divergence. Hence, results that rely on the chain rule (such as Proposition~\ref{pr:upper_bound}) fail to hold in this more general setup.

\section{Proofs of Section~\ref{sec2}}

The following proof is an extension of results presented in Ref.~\cite{rassouli2018latent}.

\begin{proof}[Proof of Proposition~\ref{pr:upper_bound}]
\label{app:upper_bound}
Let $j\in[n]$ be an arbitrary index. First note that,
\begin{align}
I(Y;\bm X |V) &= I(Y;\bm X^{\alpha_j}|V) + I(Y;\bm X^{-\alpha_j}|V, \bm X^{\alpha_j}) \\
&\geq I(Y;\bm X^{\alpha_j}|V) \\
&=  I(Y,V; \bm X^{\alpha_j} )~.
\end{align}
where the second equality is due to the independence between $V$ and $\bm
X^{\alpha_j}$. Then, we find that
\begin{align}
I(Y;V)=&\:I(Y;\bm X)-I(Y;\bm X|V)\label{e1}\\
\leq &\: I(Y; \bm X^{-\alpha_j}|\bm X^{\alpha_j})+I(Y; \bm X^{\alpha_j})-I(Y,V; \bm X^{\alpha_j} )\nonumber\\
=&\: I(Y; \bm X^{-\alpha_j}|\bm X^{\alpha_j}) - I(V;\bm X^{\alpha_j}|Y) \nonumber\\
\leq &\: I(Y; \bm X^{-\alpha_j}|\bm X^{\alpha_j}) ,\label{fe2}
\end{align}
where (\ref{e1}) follows from the Markov chain $Y-\bm X-V$. Since $j$ is chosen
arbitrarily, (\ref{fe2}) holds for all $j\in[n]$, resulting in
(\ref{uppbound}).

The inequalities in the above derivation turn into an equality when
\begin{equation}
I(Y;\bm X^{-\alpha_j}|V, \bm X^{\alpha_j}) = I(V;\bm X^{\alpha_j}|Y) = 0~.
\end{equation}
\end{proof}

\section{Further properties of synergistic disclosure}
\label{app:further_props}

We first consider the convexity of the synergy over channels, i.e.  conditional probabilities relating sources $\bX$ and target $Y$. 
\begin{lemma}[Convexity of $S^{\alpha}$ over target channels] 
$S^\mathbf{\alpha}(\mathbf{X}\rightarrow Y)$ is a convex function of $p_{Y|\mathbf X}$ for a given $p_{\bf X}$.
\end{lemma}
\begin{proof}
Let us denote the maximiser of $S^\mathbf{\alpha}(\mathbf{X}\rightarrow Y)$ by $p^{\bf{\alpha}*}_{V|\mathbf X}$ (c.f. the corresponding discussion below \eqref{eq:def_syn}), and consider $p_{Y|\mathbf{X}}=\theta p^1_{Y|\mathbf X}+(1-\theta)p^2_{Y|\mathbf X}$ for $\theta\in[0,1]$. From the convexity of mutual information, we have
\begin{equation*}
 S^\mathbf{\alpha}(\mathbf{X}\rightarrow Y)\leq \theta I_1(V;Y) + (1-\theta)I_2(V;Y),   
\end{equation*}
where $I_1$, and $I_2$ are evaluated over %
\begin{align}
    p_1(v,y)&=\sum_{\bf X}p^{\bf{\alpha}*}_{V|\mathbf X}\cdot p_{\bf X}\cdot p^1_{Y|\mathbf X},\text{ and}\\ 
    p_2(v,y)&=\sum_{\bf X}p^{\bf{\alpha}*}_{V|\mathbf X}\cdot p_{\bf X}\cdot p^2_{Y|\mathbf X}, 
\end{align}
respectively. It is also readily verified that
\begin{equation*}
   p_1(v|\mathbf x)=p_2(v|\mathbf x)=p^{\mathbf{\alpha}*}(v|\mathbf x)\in \mathcal{C}(\bm X; \bm \alpha).
\end{equation*}
As a result, one finds that
\begin{equation*}
 S^\mathbf{\alpha}(\mathbf{X}\rightarrow Y)\leq \theta S^\mathbf{\alpha}_1(\mathbf{X}\rightarrow Y) + (1-\theta)S^\mathbf{\alpha}_2(\mathbf{X}\rightarrow Y).   
\end{equation*}
\end{proof}

We now considered an extension of the data processing inequality
of the mutual information to the case of $\bm\alpha$ synergies.

\begin{lemma}[Data processing inequality for $\bm\alpha$-synergy] \label{lemma:DPI}
If $\bX - Y - Z$ is a Markov chain, then
\begin{equation}
S^{\bm\alpha}(\bX \rightarrow Y) \geq S^{\bm\alpha}(\bX \rightarrow Z)~.
\end{equation}
\end{lemma}
\begin{proof}
A direct calculation shows that
\begin{align}
S^{\bm \alpha}(\bX \rightarrow Y) 
&= \sup_{\substack{p_{V|X^n}\in \mathcal{C}(\bm X;\bm\alpha):\\V-\bm X-Y}} I(V;Y) \\
&\geq \sup_{\substack{p_{V|X^n}\in \mathcal{C}(\bm X;\bm\alpha):\\V-\bm X-Z}} I(V;Z) \label{eq:important}\\
&= S^{\bm \beta}(\bX \rightarrow Z)~.
\end{align}
Above, \eqref{eq:important} uses the fact that $\mathcal{C}(\bX;\bm\alpha)$ depends only on $\bX$ and not on
the target variable, and the traditional data processing inequality over the Markov Chain $V-\bX-Y-Z$.
\end{proof}

Finally, the last proposition explored here characterises conditions under which when there is no
$\bm\alpha$-synergy. The proof of this result is omitted, as it is a direct extension of Ref.~\cite[Proposition~1]{rassouli2018latent}
\begin{lemma}\label{prop:no_syn}
$S^{\bm\alpha }(\bm X \rightarrow Y) = 0$ if and only if
$\textnormal{Null}(\mathbf{P}_{\bm\alpha})\not\subset\textnormal{Null}(\mathbf{P}_{Y|\bm
X})$.
\end{lemma}
%


\section{Proofs of Section~\ref{sec:decomp}}

\begin{proof}[Proof of Lemma~\ref{lemma:antichains}] \label{app:antichains}

As per Lemma~\ref{lemma:markov}, the synergistic channel of interest depends
only on the null space of $\mathbf{P}_{\bm\alpha}$. Recall that adding a new
source $\alpha'$ to an existing source-set
$\bm\alpha=\{\alpha_1,\dots,\alpha_L\}$ corresponds to appending rows to
$\mathbf{P}_{\bm\alpha}$ (c.f. Eq.~\ref{eq:matp}). If $\alpha' \subset
\alpha_i$, the new rows added to $\mathbf{P}_{\bm\alpha}$ corresponding to
$\alpha'$ are linearly dependent on the existing rows, and therefore the null
space of the matrix (and thus the synergistic channel) remains unchanged.

From this same line of reasoning, the smallest such source-set is that in
which no source is contained in another one -- i.e. an anti-chain.

\end{proof}

\begin{proof}[Proof of Lemma~\ref{lemma:monotonic}] \label{app:Corder}

Consider $\bm\alpha,\bm\beta\in\mathcal{A}^*$ with $\bm\alpha \Corder \bm\beta$. 
Then, it is direct to check that $\mathcal{C}(\bX; \bm\beta) \subseteq \mathcal{C}(\bX; \bm\alpha)$,
and therefore
\begin{align}
S^{\bm \alpha}(\bX \rightarrow Y) 
&= \sup_{\substack{p_{V|X^n}\in \mathcal{C}(\bm X;\bm\alpha):\\V-\bm X-Y}} I(V;Y) \\
&\geq \sup_{\substack{p_{V|X^n}\in \mathcal{C}(\bm X;\bm\beta):\\V-\bm X-Y}} I(V;Y) \\
&=S^{\bm \alpha}(\bX \rightarrow Y)~.
\end{align}
Above, the inequality is because the supremum is taken over a smaller set of parameters.
\end{proof}

Note that the proof of the \emph{weak monotonicity} property (c.f. Section\ref{sec:axioms}) 
follows exactly the same pattern, but leveraging the fact that 
$\mathcal{C}(\bX,\bm\alpha) \subseteq \mathcal{C}\big( (\bX,Z), \bm\alpha\big)$.
The details are left to the interested reader.

\section{Characterisation of synergistic channels in binary bivariate systems}
\label{app:bivariate}

Without loss of generality, let us consider the joint distribution of binary variables $(X_1,X_2)$ described by
\begin{equation}\label{eq:distribution_binary}
p_{X_1,X_2} = 
\begin{bmatrix}
r, a - r, b - r, 1 - a - b + r 
\end{bmatrix}~,
\end{equation}
where $\mathbb{P}\{ X_1=1 \}=a$ and $\mathbb{P}\{ X_2=1 \}=b$ with $a\geq b$
determine the marginal distributions, and $\mathbb{P}\{ X_1=1,X_2=1 \}=r\in
[0,R]$ with $R=\min\{a,b\}$ gives the interdependency (note that $X_1$ and $X_2$
are independent when $r=ab$). 

The optimal $\bm\alpha$-synergistic channel for $\bm\alpha=\{\{1\},\{2\}\}$ for 
this system has been shown to be (see Ref.~\cite{rassouli2018latent})
\begin{equation}
P_{V^*|\bX} = 
\begin{bmatrix}
\frac{r( a - R )}{R(a-r)} & 1 & \frac{r( 1 - a - b + R )}{R(1 - a - b + r)}  & \frac{r( b - R )}{R(b-r)}\\
\frac{a( R - r )}{R(a-r)} & 0 & \frac{r( 1 - a - b )( R - r )}{R(1 - a - b -r)}  & \frac{b( R - r )}{R(b-r)}
\end{bmatrix}~.
\end{equation}
Interestingly, $P_{V^*|\bX}$ depends on the distribution of $\bX$ but not on $Y$. For the particular case
of $a=b=1/2$ and $r=ab=1/4$, $P_{V^*|\bX}$ reduces to an \texttt{XOR}.

\section{Simulation details}
\label{app:simulation}

This Appendix provides simulation details for the numerical results in the paper.

\paragraph*{Simulations in Section~\ref{sec:examples}:}
For Fig.~\ref{fig:ising} we use Gibbs distributions as specified in Eq.~\eqref{eq:gibbs} 
with the Hamiltonian given by Eq.~\eqref{eq:hamiltonian}. 
In contrast, for Fig.~\figsubref{fig:ising}{b} we also considered Gibbs distributions but
this time with Hamiltonians that only have terms of order $k$, i.e.
\begin{align*}
  \mathcal{H}_k(\bm x^n) = - x_{n+1} \sum_{|\bs{\gamma}|=k} J_{\bs{\gamma}} \prod_{i\in \bs{\gamma}} x_i~.
\end{align*}
In all simulations, all interaction coefficients $J$ in the Hamiltonians are drawn 
i.i.d. from a normal distribution with zero mean and standard deviation 0.1. 
In both simulations, 25 Hamiltonians are sampled at
random for each $k$, and the mean and standard deviation of the resulting
quantities ($B^m$ or $B_\partial^m$) are reported in both panes of
Figure~\ref{fig:ising}.

\paragraph*{Simulations in Sections~\ref{sec:relationship_pid} and~\ref{sec:numerical}:} 
For these simulations, each distribution $p_{\bX}$ is sampled from a symmetric Dirichlet
distribution with concentration parameter $\alpha$. Let us define
$\mathrm{Dir}(K, \alpha)$ as the Dirichlet distribution over the $(K-1)$
simplex with all parameters $\alpha_1 = \dots = \alpha_K = \alpha$. 
Sampling from this Dirichlet with a fixed $\alpha$, however, has the
undesirable effect of generating distributions with a very narrow distribution
of entropy $H(\bX)$~\cite{nemenman2002entropy}. To generate distributions 
with a near-uniform of entropy, we sample $\alpha$ from a Nemenman-Shafee-Bialek (NSB)
prior~\cite{archer2013bayesian}
\begin{align*}
  p(\alpha) \propto K \psi(K \alpha + 1) + \psi(\alpha + 1) ~ ,
\end{align*}
for a distribution over an alphabet of size $K$. For simulations of $n$ binary
variables, we set $K = 2^n$, sample $\alpha$ using the equation above, and
then sample $p_{\bX}$ from a symmetric Dirichlet using standard algorithms.

\vspace{10pt}
\section{Asymptotic limits of self-disclosure}
\label{app:dominance}
\begin{proof}[Proof of Proposition~\ref{prop:asynto}]
From Corollary 1.2 of Ref.~\cite{Rassouli2019} one can see that
\begin{equation}
\min_{p_{V|\bX}\in \mathcal{C}(\bX;\bm\gamma_m)} H(\bX|V) \leq \log\big( \text{rank}(\mathbf{P}_{\bm\gamma_m}) \big)~.
\end{equation}
As the rank of a matrix cannot be larger than its number of rows, for a given
$m\in \{1,\dots,n-1\}$ is clear that
\begin{equation}
\text{rank}( \mathbf{P}_{\bm\gamma_m} ) \leq m \sum_{k=1}^n |\mathcal{X}_k| \leq m n K~,
\end{equation}
and therefore $\min_{p_{V|\bX}\in \mathcal{C}(\bX;\bm\gamma_m)} H(\bX|V) \leq \log(mnK)$. By
definition of $B^m(\bX)$, this implies that
\begin{equation}
H(\bX) - \log(mnK) \leq B^m(\bX) \leq H(\bX)~.
\end{equation}
Taking the limit of $n\rightarrow\infty$ for $m$ fixed gives that
\begin{equation}
\lim_{n\rightarrow\infty}  \frac{1}{n} B^m(\bX) = \lim_{n\rightarrow\infty} \frac{1}{n} H(\bX)~,
\end{equation}
which is equivalent to what we want to prove (note that the cardinality of $\bX$ 
grows with $n$ as well).
\end{proof}

\bibliographystyle{IEEEtran}
\bibliography{bib}

\begin{thebibliography}{10}
\providecommand{\url}[1]{#1}
\csname url@samestyle\endcsname
\providecommand{\newblock}{\relax}
\providecommand{\bibinfo}[2]{#2}
\providecommand{\BIBentrySTDinterwordspacing}{\spaceskip=0pt\relax}
\providecommand{\BIBentryALTinterwordstretchfactor}{4}
\providecommand{\BIBentryALTinterwordspacing}{\spaceskip=\fontdimen2\font plus
\BIBentryALTinterwordstretchfactor\fontdimen3\font minus
  \fontdimen4\font\relax}
\providecommand{\BIBforeignlanguage}[2]{{%
\expandafter\ifx\csname l@#1\endcsname\relax
\typeout{** WARNING: IEEEtran.bst: No hyphenation pattern has been}%
\typeout{** loaded for the language `#1'. Using the pattern for}%
\typeout{** the default language instead.}%
\else
\language=\csname l@#1\endcsname
\fi
#2}}
\providecommand{\BIBdecl}{\relax}
\BIBdecl

\bibitem{ganmor2011sparse}
E.~Ganmor, R.~Segev, and E.~Schneidman, ``Sparse low-order interaction network
  underlies a highly correlated and learnable neural population code,''
  \emph{Proceedings of the National Academy of Sciences}, vol. 108, no.~23, pp.
  9679--9684, 2011.

\bibitem{wibral2017quantifying}
M.~Wibral, C.~Finn, P.~Wollstadt, J.~T. Lizier, and V.~Priesemann,
  ``Quantifying information modification in developing neural networks via
  partial information decomposition,'' \emph{Entropy}, vol.~19, no.~9, 2017.

\bibitem{tax2017partial}
T.~M. Tax, P.~A.~M. Mediano, and M.~Shanahan, ``The partial information
  decomposition of generative neural network models,'' \emph{Entropy}, vol.~19,
  no.~9, 2017.

\bibitem{rosas2018selforg}
F.~Rosas, P.~A. Mediano, M.~Ugarte, and H.~J. Jensen, ``An
  information-theoretic approach to self-organisation: Emergence of complex
  interdependencies in coupled dynamical systems,'' \emph{Entropy}, vol.~20,
  no.~10, 2018.

\bibitem{rosas2019quantifying}
F.~E. Rosas, P.~A.~M. Mediano, M.~Gastpar, and H.~J. Jensen, ``Quantifying
  high-order interdependencies via multivariate extensions of the mutual
  information,'' \emph{Physical Review E}, vol. 100, p. 032305, Sep 2019.

\bibitem{waldrop1993complexity}
M.~M. Waldrop, \emph{Complexity: The Emerging Science at the Edge of Order and
  Chaos}.\hskip 1em plus 0.5em minus 0.4em\relax Simon and Schuster, 1993.

\bibitem{chechik2002group}
G.~Chechik, A.~Globerson, M.~J. Anderson, E.~D. Young, I.~Nelken, and
  N.~Tishby, ``Group redundancy measures reveal redundancy reduction in the
  auditory pathway,'' in \emph{Advances in Neural Information Processing
  Systems}, 2002, pp. 173--180.

\bibitem{varadan2006computational}
V.~Varadan, D.~M. Miller~III, and D.~Anastassiou, ``Computational inference of
  the molecular logic for synaptic connectivity in {C. elegans},''
  \emph{Bioinformatics}, vol.~22, no.~14, pp. e497--e506, 2006.

\bibitem{barrett2015exploration}
A.~B. Barrett, ``Exploration of synergistic and redundant information sharing
  in static and dynamical gaussian systems,'' \emph{Physical Review E},
  vol.~91, p. 052802, May 2015.

\bibitem{amari2001information}
S.-I. Amari, ``Information geometry on hierarchy of probability
  distributions,'' \emph{IEEE Transactions on Information Theory}, vol.~47,
  no.~5, pp. 1701--1711, 2001.

\bibitem{schneidman2003network}
E.~Schneidman, S.~Still, M.~J. Berry, and W.~Bialek, ``Network information and
  connected correlations,'' \emph{Physical Review Letters}, vol.~91, no.~23, p.
  238701, 2003.

\bibitem{williams2010nonnegative}
P.~L. Williams and R.~D. Beer, ``Nonnegative decomposition of multivariate
  information,'' \emph{arXiv preprint arXiv:1004.2515}, 2010.

\bibitem{Note1}
In effect, PID merely states formal relationships between its atoms and various
  Shannon's mutual information terms.

\bibitem{ince2017measuring}
R.~A. Ince, ``Measuring multivariate redundant information with pointwise
  common change in surprisal,'' \emph{Entropy}, vol.~19, no.~7, p. 318, 2017.

\bibitem{bertschinger2014quantifying}
N.~Bertschinger, J.~Rauh, E.~Olbrich, J.~Jost, and N.~Ay, ``Quantifying unique
  information,'' \emph{Entropy}, vol.~16, no.~4, pp. 2161--2183, 2014.

\bibitem{finn2018pointwise}
C.~Finn and J.~T. Lizier, ``Pointwise partial information decomposition using
  the specificity and ambiguity lattices,'' \emph{Entropy}, vol.~20, no.~4, p.
  297, 2018.

\bibitem{james2019unique}
R.~James, J.~Emenheiser, and J.~Crutchfield, ``Unique information and secret
  key agreement,'' \emph{Entropy}, vol.~21, no.~1, p.~12, 2019.

\bibitem{thurner2018introduction}
S.~Thurner, R.~Hanel, and P.~Klimek, \emph{Introduction to the Theory of
  Complex Systems}.\hskip 1em plus 0.5em minus 0.4em\relax Oxford University
  Press, 2018.

\bibitem{rauh2014reconsidering}
J.~Rauh, N.~Bertschinger, E.~Olbrich, and J.~Jost, ``Reconsidering unique
  information: Towards a multivariate information decomposition,'' in
  \emph{2014 IEEE International Symposium on Information Theory}.\hskip 1em
  plus 0.5em minus 0.4em\relax IEEE, 2014, pp. 2232--2236.

\bibitem{kolchinsky2019novel}
A.~Kolchinsky, ``A novel approach to multivariate redundancy and synergy,''
  \emph{arXiv preprint arXiv:1908.08642}, 2019.

\bibitem{feldman1998measures}
D.~P. Feldman and J.~P. Crutchfield, ``Measures of statistical complexity:
  Why?'' \emph{Physics Letters A}, vol. 238, no. 4-5, pp. 244--252, 1998.

\bibitem{banerjee2018computing}
P.~K. Banerjee, J.~Rauh, and G.~Mont{\'u}far, ``Computing the unique
  information,'' in \emph{2018 IEEE International Symposium on Information
  Theory (ISIT)}.\hskip 1em plus 0.5em minus 0.4em\relax IEEE, 2018, pp.
  141--145.

\bibitem{Note2}
The direct implications of these approaches are about the unique information,
  and apply to the synergy only via additional equalities with mutual
  information terms.

\bibitem{rassouli2018latent}
B.~Rassouli, F.~Rosas, and D.~G{\"u}nd{\"u}z, ``Latent feature disclosure under
  perfect sample privacy,'' in \emph{2018 IEEE International Workshop on
  Information Forensics and Security (WIFS)}.\hskip 1em plus 0.5em minus
  0.4em\relax IEEE, 2018, pp. 1--7.

\bibitem{Rassouli2019}
B.~{Rassouli}, F.~E. {Rosas}, and D.~G{\"u}nd{\"u}z, ``Data disclosure under
  perfect sample privacy,'' \emph{IEEE Transactions on Information Forensics
  and Security}, pp. 1--1, 2019.

\bibitem{quax2017quantifying}
R.~Quax, O.~Har-Shemesh, and P.~Sloot, ``Quantifying synergistic information
  using intermediate stochastic variables,'' \emph{Entropy}, vol.~19, no.~2,
  p.~85, 2017.

\bibitem{Note3}
To verify this, first one shows that it suffices to have a random variable $V$
  with a finite alphabet by means of cardinality bounding techniques. Then,
  using the fact that any finite probability simplex is a compact set, the
  supremum in (\ref {eq:def_syn}) has to be attained due to the continuity of
  the mutual information.

\bibitem{Note4}
Shannon himself employed the known formula for entropy because it had a strict
  operational meaning as minimum description length, not because it was the
  result of the four celebrated axioms.

\bibitem{Note5}
In fact, we employ the notation $\protect \mathcal {A}^*$ to differentiate this
  set from with their definition of antichains set $\protect \mathcal {A}$ that
  doesn't include the empty set -- which plays an important role in our
  framework.

\bibitem{james2018unique}
R.~James, J.~Emenheiser, and J.~Crutchfield, ``Unique information via
  dependency constraints,'' \emph{Journal of Physics A: Mathematical and
  Theoretical}, 2018.

\bibitem{ay2019information}
N.~Ay, D.~Polani, and N.~Virgo, ``Information decomposition based on
  cooperative game theory,'' \emph{arXiv preprint arXiv:1910.05979}, 2019.

\bibitem{charalambides2018enumerative}
C.~A. Charalambides, \emph{Enumerative Combinatorics}.\hskip 1em plus 0.5em
  minus 0.4em\relax Chapman and Hall/CRC, 2018.

\bibitem{Note6}
Strictly speaking, since $\protect \bm {\alpha }$ is a set of sets, these atoms
  should be denoted by $S^{\protect \{\emptyset \protect \}}_\partial $,
  $S^{\protect \{\protect \{1\protect \}\protect \{2\protect \}\protect
  \}}_\partial $, etc. For clarity, and for consistency with prior literature,
  we omit the outer bracket and denote these symbols by shortened expressions
  (e.g. $S^{\emptyset }_\partial $, $S^{\protect \{1\protect \}\protect
  \{2\protect \}}_\partial $, etc.).

\bibitem{Note999}
Note that since $S^{\protect \{12\protect \}} = 0$, we have $S_\partial
  ^{\protect \{1\protect \}\protect \{2\protect \}} = S^{\protect \{1\protect
  \}\protect \{2\protect \}}$.

\bibitem{Note7}
Proofs of this result can be found in Refs.~\cite {rassouli2018latent}
  and~\cite {quax2017quantifying}.

\bibitem{Note8}
For more information on this relationship, please see Refs.~\cite
  [Lemma~3]{rosas2019quantifying} and \cite [Section~VI]{rosas2018selforg}.

\bibitem{Note9}
Note that the presented findings consider a particular ensemble of
  distributions. Therefore, more work is needed in order to explore their
  generality.

\bibitem{griffith2014quantifying}
V.~Griffith and C.~Koch, ``Quantifying synergistic mutual information,'' in
  \emph{Guided Self-Organization: Inception}.\hskip 1em plus 0.5em minus
  0.4em\relax Springer, 2014, pp. 159--190.

\bibitem{harder2013bivariate}
M.~Harder, C.~Salge, and D.~Polani, ``Bivariate measure of redundant
  information,'' \emph{Physical Review E}, vol.~87, no.~1, p. 012130, 2013.

\bibitem{Note10}
Since our framework provides a algorithm to build the optimal self-synergistic
  channel for arbitrary sources $\protect \ensuremath {\protect \bm
  {X}}\protect \xspace $, it would be natural to conjecture that this channel
  could also be optimal for other target variables -- i.e., that it could serve
  as sufficient statistic under the corresponding constraints. Unfortunately,
  numerical explorations show that, while this works for two binary variables
  (Appendix~\ref {app:bivariate}), it is in general not the case.

\bibitem{griffith2014intersection}
V.~Griffith, E.~Chong, R.~James, C.~Ellison, and J.~Crutchfield, ``Intersection
  information based on common randomness,'' \emph{Entropy}, vol.~16, no.~4, pp.
  1985--2000, 2014.

\bibitem{Note11}
As an example of this, if $(X_1,X_2)$ are two independent fair coins and
  $Y=(X_1,X_2)$, then a direct calculation shows that, if $\bm \alpha
  =\{\{1\},\{2\}\}$, then $S^{\bm \alpha }\big ((X_1,X_2)\rightarrow Y\big ) =
  1$ and $S^{\bm \alpha }\big ((X_1,Y)\rightarrow X_2\big ) = 0$.

\bibitem{Note12}
Consider a ``double-\protect \texttt {XOR}'' distribution, with 3 independent
  bits as inputs, and $Y = (X_1 \texttt {xor} X_2, X_2 \texttt {xor} X_3)$ as
  output. For this distribution, all atoms $S_\partial ^{\{ij\}\{k\}}(\bX
  \rightarrow Y) = -1$, violating (\protect \textbf {LP}). To see why, note
  that $S^{\{12\}\{13\}}(\bX \rightarrow Y) = S^{\{12\}\{3\}}(\bX \rightarrow
  Y) = 1$, since in both cases $X_2 \texttt {xor} X_3$ can be disclosed,
  yielding the second bit of $Y$; and $S^{\{12\}\{23\}}(\bX \rightarrow Y) =
  1$, since $X_1 \texttt {xor} X_3$ can be disclosed, yielding the parity of
  $Y$. Hence, applying the M{\"o}bius inversion, we have $S_\partial
  ^{\{12\}\{3\}}(\bX \rightarrow Y) = -1$.

\bibitem{james2018dit}
R.~G. James, C.~J. Ellison, and J.~P. Crutchfield, ``{dit}: a {P}ython package
  for discrete information theory,'' \emph{The Journal of Open Source
  Software}, vol.~3, no.~25, p. 738, 2018.

\bibitem{Note13}
We do not take a stance here with respect to the non-negativity of information
  atoms; but since the atoms have to sum to the same mutual information,
  negative values in the lower atoms necessarily entail inflated synergy
  values.

\bibitem{Note14}
A Python implementation of synergistic disclosure and the corresponding
  decomposition can be found online at \protect \texttt {\protect \url
  {https://github.com/pmediano/syndisc}}.

\bibitem{nemenman2002entropy}
I.~Nemenman, F.~Shafee, and W.~Bialek, ``Entropy and inference, revisited,'' in
  \emph{Advances in Neural Information Processing Systems}, 2002, pp. 471--478.

\bibitem{archer2013bayesian}
E.~Archer, I.~Park, and J.~Pillow, ``Bayesian and quasi-{Bayesian} estimators
  for mutual information from discrete data,'' \emph{Entropy}, vol.~15, no.~5,
  pp. 1738--1755, 2013.

\end{thebibliography}

\end{document}